\documentclass[12pt]{article}
\usepackage{amsmath,amsthm,comment}
\usepackage{algorithm}
\usepackage{algorithmic}
\usepackage{amsfonts}
\usepackage{multirow}
\usepackage{booktabs}
\usepackage{graphicx}
\usepackage{epstopdf}
\usepackage{xcolor}
\usepackage[title]{appendix}
\usepackage[top=1in, bottom=1in, left=1in, right=1in]{geometry}
\usepackage{amssymb}
\usepackage{stmaryrd}
\usepackage{cleveref}
\usepackage{graphicx} 
\usepackage{caption}
\usepackage{mwe}
\usepackage{graphics,graphicx,epsfig,psfrag,subfigure,threeparttable}
\usepackage{float}
\usepackage{algorithm}
\usepackage{algorithmic}

\newtheorem{theorem}{Theorem}
\newtheorem{definition}{Definition}
\newtheorem{lemma}{Lemma}
\usepackage{color}
\usepackage{authblk}
\usepackage{url}

\title{One-Bit Matrix Completion with Differential Privacy}
\author[1]{\normalsize Zhengpin Li}
\author[1]{Zheng Wei}
\author[1]{Zengfeng Huang}
\author[2]{Xiaojun Mao}
\author[1]{Jian Wang \thanks{Zhengpin Li and Zheng Wei contributed equally to this work, and Xiaojun Mao and Jian Wang are the co-corresponding authors. (E-mails: maoxj@sjtu.edu.cn, jian\_wang@fudan.edu.cn)}}
\affil[1]{School of Data Science, Fudan University, China }
\affil[2]{School of Mathematical Sciences, Shanghai Jiao Tong University, China}
\date{}
\begin{document}
\maketitle

\begin{abstract}
As a prevailing collaborative filtering method for recommendation systems, one-bit matrix completion requires data collected by users to provide personalized service. Due to insidious attacks and unexpected inference, the release of users' data often raises serious privacy concerns. To address this issue, differential privacy~(DP)~has been widely used in standard matrix completion models. To date, however, little has been known about how to apply DP to achieve privacy protection in one-bit matrix completion. In this paper, we propose a unified framework for ensuring a strong privacy guarantee of one-bit matrix completion with DP. In our framework, we develop four different private perturbation mechanisms corresponding to different stages of one-bit matrix completion. For each mechanism, we design a privacy-preserving algorithm and provide a theoretical recovery error bound under the proper conditions. Numerical experiments on synthetic and real-world datasets demonstrate the effectiveness of our proposal. Compared to the one-bit matrix completion without privacy protection, our proposed mechanisms can maintain high-level privacy protection with marginal loss of completion accuracy.
\end{abstract}

\section{Introduction}
With the dramatically increasing rise of the data volume collected in online markets, recommendation systems have gained a great deal of interest in many realms for their capacity in providing personalized service, precisely and timely~\cite{he2017neural, davidson2010youtube}. In the early stage, recommendation systems focused primarily on leveraging the historical rating records of users to items directly. They made the corresponding recommendations by selecting the most similar objects from the recorded candidates~\cite{sarwar2001item}. However, those methods neglect the underlying interaction information between users and items that is potentially effective and even critical for improving the recommendation performance. To overcome this shortcoming, approaches based on collaborative filtering have been proposed and gradually become mainstream for recommendation tasks~\cite{candes2010matrix, xue2019deep}.
 
As the most commonly used collaborative filtering method, matrix completion has achieved great success in many real-life scenarios~\cite{koren2009matrix}. Matrix completion captures the underlying information between users and items by decomposing the original matrix into two low-rank ones, which represent the users' and items' implicit preferences, respectively. However, in some cases where matrices are comprised of binary values only, standard matrix completion cannot be applied directly. To address this issue, Davenport~{\it et al.}~\cite{davenport20141}~first proposed the one-bit matrix completion and analyzed the recovery error bound. By formulating the problem to a nuclear-norm constrained maximum likelihood form, they solved it through the non-monotonic spectral projected gradient~(SPG) algorithm. There have also been other works focused on one-bit matrix completion problems under various circumstances~\cite{cai2013max,bhaskar20151, ni2016optimal}, which provide useful tools for recommendation systems.

As recommendation systems require the personal preference information of the users, they often raise serious privacy concerns for matrix completion. For example, the users' demographic information, such as the age, gender, and even political orientation, can be inferred from the interaction information of users and items~(e.g., seeing a movie or purchasing a product)~\cite{weinsberg2012blurme}. To tackle this issue, a variety of privacy-preserving methods have been proposed based on the definition of differential privacy~(DP)~\cite{dwork2006calibrating}. DP has emerged as a general standard for measuring the level of privacy protection for matrix completion because it has no requirement of the adversary's background knowledge while enjoying the merits of easy computation and provable privacy guarantees.

Unfortunately, even though one-bit matrix completion works well for completion tasks in general, there have been little research when it comes to privacy protection. Based on the one-bit matrix completion approach presented in~\cite{davenport20141}, we propose in this paper a unified framework for privacy-preserving one-bit matrix completion via differential privacy. Our framework extends the application of differential privacy from the standard matrix completion to one-bit matrix completion. Corresponding to the different steps during the completion process, we design four perturbation mechanisms: i) input, ii) objective, iii) gradient, and iv) output perturbation.

Unlike the standard matrix completion, input perturbation mechanisms for one-bit matrix completion cannot be implemented by adding noises to the entries directly. This is because the binary values only represent the category in essence, rather than the magnitudes. To address this challenge, we design a novel probabilistic approach for the private input mechanism, which can change the original entries to their counterparts with a certain probability. For the other three perturbation mechanisms, we add noises to the underlying preference matrix that consists of continuous-valued entries, thus impacting the observed matrix to achieve privacy protection. We establish a provable privacy-preserving guarantee for each of the four mechanisms. Furthermore, we provide the upper bound of recovery error under appropriate conditions for input, objective, and output perturbation. To the best of our knowledge, we are the first to propose a privacy-preserving one-bit matrix completion via differential privacy. Experimental results indicate that our proposal yields high completion accuracy while ensuring strong privacy protections.

The main contributions of our paper are summarized as follows.

\begin{itemize}
   \item We propose a novel framework for applying differential privacy to one-bit matrix completion. In our framework, four specific perturbation mechanisms are developed corresponding to different stages of the completion process. We provide complete procedures for privacy-preserving algorithms and rigorous guarantees of $\epsilon$-differential privacy for these mechanisms.
   
   \item We establish upper bounds of the recovery error for the proposed mechanisms under certain conditions, which demonstrate the privacy-accuracy trade-offs.

   \item We show the effectiveness of the proposed framework on both synthetic and real-world datasets. Our results show that the proposed privacy-preserving approaches can ensure strong privacy guarantees with a slight loss of completion accuracy.
\end{itemize}

The rest of this paper is organized as follows. We overview related work on one-bit matrix completion and differential privacy in Section~\ref{sec2}, followed by the notations and preliminaries in Section~\ref{secpre}. We introduce our framework and theoretical analysis in Section~\ref{secmethod}. Finally, we present experimental results and conclusion remarks in Sections~\ref{secsim} and~\ref{seccon}, respectively.

\section{Related Work}\label{sec2}
This section introduces the existing research on one-bit matrix completion and differential privacy.

\subsection{One-Bit Matrix Completion}
Matrix completion aims at recovering the unobserved entries of a matrix. It has been used in a variety of practical applications, including collaborative filtering~\cite{koren2009matrix}, sensor localization~\cite{liu2010interior}, and compressed sensing~\cite{candes2006compressive,jacques2013robust}. The theories for matrix completion have been well studied. Cand{\`e}s {\it et al.}~\cite{candes2010power, candes2009exact, candes2010matrix}~established many pioneering theories in this area. Not only did they demonstrate conditions in terms of the number of observations required to recover a low-rank matrix perfectly, but also they showed that matrix completion can achieve accurate performance, even when a few observations have been corrupted by noises. In particular, they rigorously proved that for a matrix of size $n \times n$ and rank $r$, samples in the order of $nr\text{polylog}(n)~$ are required for perfect recovery. Considering a situation where a few observed entries are corrupted with noise, about $nr \log^2(n)$ samples 
can guarantee recovery with an error that is proportional to the noise level. Keshavan {\it et al.}~\cite{keshavan2010matrix}~and Gross~\cite{gross2011recovering}~also established results on the theoretical perspective of low-rank matrices completion. After that, enormous works about the variants of low-rank matrix completion or correlative optimization methods were presented~\cite{negahban2012restricted, gaiffas2017high, koltchinskii2011nuclear,  rohde2011estimation}.

Davenport {\it et al.}~\cite{davenport20141}~first extended the theory of matrix completion to the one-bit case, where the observations of the matrix consist of binary values. Specifically, they proposed a trace-norm constrained maximum likelihood estimator and proved that their proposal minimax rate optimal under the uniform sampling model. Cai and Zhou~\cite{cai2013max}~presented a maximum likelihood estimator under the constraint of max-norm for the case where observations are sampled in a general sampling model. They also proved that the estimation error bound is rate-optimal, as it matches the minimax lower bound established under the general sampling model. Bhaskar and Javanmard~\cite{bhaskar20151}~proposed a maximum likelihood estimator that satisfies an infinite-norm constraint and provided an upper bound for the recoverycompletion error. They also developed an iterative algorithm for solving nonconvex optimization problems and established the proof of global optimality at the limiting point. Ni and Gu~\cite{ni2016optimal}~developed a unified framework for solving one-bit matrix completion with rank constraints under a general random sampling model. The method integrates the statistical error of the estimator and the optimization error of the algorithm, and outperforms the estimators with max-norm and nuclear-norm constraints.

\subsection{Differential Privacy}

There have been many perturbation mechanisms for achieving differential privacy. In general, these mechanisms include four major types: input perturbation, objective perturbation, gradient perturbation, and output perturbation. These perturbation methods protect privacy by adding noises in different stages of data processing. Specifically, Dwork~{\it et al.}~\cite{dwork2006calibrating} first proposed the idea of input perturbation techniques, where the underlying data are randomly modified, and answers to questions are computed using the modified data. Afterward, Friedman~{\it et al.}~\cite{friedman2016differential} proposed to utilize the input perturbation in matrix completion for privacy protection. They assumed that each entry is independent of the rest, and maintained DP by adding noises satisfying specific exponential distribution to the input matrix.

In terms of the objective perturbation, Chaudhuri and Monteleoni~\cite{chaudhuri2008privacy}~developed a privacy-preserving regularized logistic regression algorithm by solving a perturbed optimization problem. Moreover, Chaudhuri~{\it et al.}~\cite{chaudhuri2011differentially}~further proposed a general objective perturbation technique by adding noises to the regularized objective function for empirical risk minimization and established theoretical generalization bounds for this method. With regard to gradient perturbation, Bassily~{\it et al.}~\cite{bassily2014private}~demonstrated the excess risk bound of a noisy gradient descent satisfying differential privacy and developed a faster variant based on stochastic gradient descent~(SGD). Afterward, Wang~{\it et al.}~\cite{wang2017differentially}~proposed a novel gradient perturbation approach satisfying differential privacy with tighter utility upper bound and less running time. As for the output perturbation, Dwork~{\it et al.}~\cite{dwork2006calibrating} first proposed to add noises to outputs and report the noisy versions. Chaudhuri~{\it et al.}~\cite{chaudhuri2011differentially} extended this idea to produce privacy-preserving approximations of classifiers learned via empirical risk minimization. Wu~{\it et al.}~\cite{wu2017bolt} designed a novel differentially private algorithm for SGD based on the output perturbation and provided a rigorous analysis for its higher model accuracy.

\section{Preliminaries}\label{secpre}

This section introduces the notations and preliminaries about one-bit matrix completion and differential privacy throughout the paper.  

\subsection{Notations}
Matrices, vectors and scalars are denoted by boldface capital letters ($\mathbf{X}$), boldface lowercase letters ($\mathbf{x}$) and vanilla lowercase letters ($x$), respectively. $[d] = \{1,2,\cdots, d \}$. $\mathbf{X}_{ij}$ is the ($i,j$)-th entry of $\mathbf{X}$. $\|\mathbf{X}\|_{*}$, $\|\mathbf{X}\|_{\infty}=\max _{i, j}|\mathbf{X}_{ij}|$ and $\|\mathbf{X}\|_{F}=\sqrt{\left(\sum_{i, j}\mathbf{X}_{ij}^2\right)}$ represent the nuclear norm~(sum of the singular values), entry-wise infinity norm and Frobenius norm of $\mathbf{X}$, respectively. For vector $\mathbf{x}$, we denote its $p$-norm by $\|\mathbf{x}\|_p$. Considering a function $f(x):\mathbb{R}\mapsto \mathbb{R}$, the values
	\begin{equation*}
		L_{f,\alpha}\hspace{-.5mm}=\hspace{-.5mm}\sup _{|x| \leq  \alpha} \frac{\left|f^{\prime}(x)\right|}{f(x)(1\hspace{-.5mm}-\hspace{-.5mm}f(x))}~\text{and}~ \beta_{f,\alpha}\hspace{-.5mm}=\hspace{-.5mm}\sup _{|x| \leq  \alpha} \frac{f(x)(1\hspace{-.5mm}-\hspace{-.5mm}f(x))}{\left(f^{\prime}(x)\right)^2}
	\end{equation*}
describe the steepness and flatness of function $f(x)$, respectively. The Hellinger distance for binary scalars $x,y \in \{0,1\}$ is given by
\begin{equation*}
    \operatorname{d}_{H}^{2}(x, y):=(\sqrt{x}-\sqrt{y})^{2}+(\sqrt{1-x}-\sqrt{1-y})^{2}.
\end{equation*}
We allow the Hellinger distance to act on binary matrices $\mathbf{X}, \mathbf{Y} \in \mathbb{R}^{d_1\times d_2}$ via the average Hellinger distance:
\begin{equation*}
    \operatorname{d}_{H}^{2}(\mathbf{X}, \mathbf{Y})=\frac{1}{d_{1} d_{2}} \sum_{i, j} \operatorname{d}_{H}^{2}\left(\mathbf{X}_{ij}, \mathbf{Y}_{ij}\right).
\end{equation*}

\subsection{One-Bit Matrix Completion}\label{3.2}
One-bit matrix completion aims to recover the underlying low-rank matrix based on a limited number of observations. Instead of observing the actual entries from the underlying matrix $\mathbf{M}\in \mathbb{R}^{d_1\times d_2}$, we only have access to the signs of the noisy entries of $\mathbf{M}$. Specifically, given the underlying matrix $\mathbf{M}\in \mathbb{R}^{d_1\times d_2}$ with at most rank-$r$ and a probability function $h(x):\mathbb{R}\mapsto [0,1]$, the entries of $\mathbf{Y}$ depend on $\mathbf{M}$ in the following way,
\begin{equation}\label{generateYY}
\mathbf{Y}_{ij}=\left\{\begin{array}{ll}+1 & \text { w.p. } h\left(\mathbf{M}_{ij}\right) \\ -1 & \text { w.p. } 1-h\left(\mathbf{M}_{ij}\right)\end{array} \quad \text { for }(i,j) \in \Omega\right.,
\end{equation}
where w.p. means "with probability", $\Omega=\{\left(i_{1}, j_{1}\right), \ldots,(i_{n}, j_{n})\}$ represents indices and $n$ denotes the number of $\Omega$. Define $\mathbf{Z}\in\mathbb{R}^{d_1\times d_2}$ as a matrix with entries drawn \emph{i.i.d.} from a distribution whose cumulative distribution function satisfies $\mathbb{P}(\mathbf{Z}_{ij} \leqslant x)=1-h(-x)$. Then, we can rewrite~\eqref{generateYY} as
\begin{equation*}
\mathbf{Y}_{ij}=\left\{\begin{array}{ll}+1 & \text { if } \mathbf{M}_{ij}+\mathbf{Z}_{ij} \ge 0 \\ -1 & \text { if } \mathbf{M}_{ij}+\mathbf{Z}_{ij} \le 0\end{array} \quad \text { for }(i,j) \in \Omega\right..
\end{equation*}

Following~\cite{davenport20141, cai2013max}, we consider two natural choices for $h(x)$ or for variable $\mathbf{Z}_{ij}$. 
\begin{itemize}
	\item Logistic noise model: $h(x)=1 /\left(1+e^{-x} \right)$, which is equivalent to the fact that the noise $\mathbf{Z}_{ij}$ follows the standard logistic distribution.
	
	\item Gaussian noise model: $h(x)=1-\Phi(-x / \sigma)=\Phi(x / \sigma)$ and variance $\sigma>0$ where $\Phi(x)$ is the cumulative distribution function of a standard Gaussian. This is equivalent to that the noise $\mathbf{Z}_{ij}$ follows the standard Gaussian distribution.
\end{itemize}

To obtain an approximate $\widehat{\mathbf{M}}$, we utilize the observations to minimize the negative log-likelihood function with respect to the optimization variable $\mathbf{X}$ under a set of convex constraints. Specifically, this optimization problem can be expressed as 
\begin{equation}\label{OBj}
\underset{\mathbf{X}}{\min }~ \mathcal{L}_{\Omega, \mathbf{Y}}(\mathbf{X}), \quad \text{s.t.} ~~ \mathbf{X}\in \mathcal{C},
\end{equation}
where $\mathcal{C} = \{\mathbf{X}: \|\mathbf{X}\|_{*} \leq  \tau, \|\mathbf{X}\|_{\infty} \leq  \alpha\}$ and   
\begin{equation*}\label{Eq1} 
  \begin{aligned} 
    \mathcal{L}_{\Omega, \mathbf{Y}}\left(\mathbf{X}\right) = -\frac{1}{2}&\sum_{(i,j) \in \Omega}\left\{{\left(1 + \mathbf{Y}_{ij}\right)\log[h(\mathbf{X}_{ij})]}\right.\\
    &\left.{+ \left(1-\mathbf{Y}_{ij}\right)(\log[1-h(\mathbf{X}_{ij})])}\right\}.
  \end{aligned}
\end{equation*}
Motivated by~\cite{davenport20141}, we introduce the constraint $\|\mathbf{X}\|_{*} \leq  \tau$ as the relaxation of the low-rank constraint. The constraint $\|\mathbf{X}\|_{\infty} \leq  \alpha$ can enforce $\mathbf{X}$ to be not too spiky, thus making the problem well-posed~\cite{negahban2012restricted}.

\subsection{Differential Privacy}
This section introduces the definitions related to differential privacy and presents lemmas that will be used in the following sections.
\begin{definition}[$\epsilon$-Differential Privacy]
	A (randomized) algorithm $\mathcal{A} : \mathcal{D} \mapsto \mathbb{R}^{d}$ is said to be $\epsilon$-differential private if for all subsets $S \subseteq \mathbb{R}^{d}$ and for all datasets $\mathbf{X}, \mathbf{X}^{\prime} \in \mathcal{D}$ that differ in one entry at most, it holds that 
	\begin{equation*}
	\mathbb{P}(\mathcal{A}(\mathbf{X}) \in S) \leq  e^{\epsilon} \mathbb{P}\left(\mathcal{A}\left(\mathbf{X}^{\prime}\right) \in S\right).
	\end{equation*}
\end{definition}
The definition illustrates that the output of an $\epsilon$-differential private algorithm $\mathcal{A}$ doesn't change much if we change one entry in a database. In this way, one cannot acquire accurate information by observing the change in the output only. The parameter $\epsilon$ controls the maximum amount of the leaked information, which is called privacy leakage~\cite{jingyu2015differentially}. We now introduce the definitions of Laplace distribution, $L_p$-sensitivity, and a Laplace mechanism that ensures $\epsilon$-differential privacy.
\begin{definition}[Laplace Distribution]
	The probability density function of Laplace distribution with mean $0$ and scale $b$ is
	\begin{equation*}
		p_b(x) = \frac{1}{2b}\exp\left(-\frac{|x|}{b}\right).
	\end{equation*}
\end{definition}

\begin{definition}[$L_p$-Sensitivity]
	The $L_{p}$-sensitivity of an algorithm $\mathcal{A}: \mathcal{D} \mapsto \mathbb{R}^{d}$ is the smallest number $\Delta_p(\mathcal{A})$ such that for all $\mathbf{X}, \mathbf{X}^{\prime} \in \mathcal{D}$ which differ in a single entry,
	\begin{equation*}
	\left\|\mathcal{A}(\mathbf{X})-\mathcal{A}\left(\mathbf{X}^{\prime}\right)\right\|_{p} \leq  \Delta_p(\mathcal{A}).
	\end{equation*}
\end{definition}
\begin{lemma}[Laplace Mechanism~\cite{dwork2006calibrating}]\label{Lap}
	For any algorithm $\mathcal{A}: \mathcal{D} \mapsto \mathbb{R}^{d}$ and $\mathbf{X}\in \mathcal{D}$, let $\{t_i\}_{i=1,2,\cdots,d}$ be {\it i.i.d.} variables drawn from Laplace distribution with mean $0$ and scale $\Delta_1(\mathcal{A})/\epsilon$, the Laplace mechanism $\text{Lap}_\mathcal{A}(\cdot)$ is defined as
	\begin{equation*}
	\operatorname{Lap}_{\mathcal{A}}(\mathbf{X})=\mathcal{A}(\mathbf{X})+\left(t_{1}, \ldots, t_{d}\right)^\top,
	\end{equation*}
	which ensures $\epsilon$-differential privacy.
\end{lemma}
The Laplace mechanism adds random noises to the numeric output of an algorithm, in which the magnitude of noises follows Laplace distribution with variance $\Delta_1(\mathcal{A})/\epsilon$, where $\Delta_{1}(\mathcal{A})$ represents the global sensitivity of $\mathcal{A}$.  We also rely on the following exponential mechanism, which is based on the $L_2$-sensitivity.

\begin{lemma}[Exponential Mechanism~\cite{wu2017bolt}]\label{exp}
	Let $\mathcal{A}: \mathcal{D} \mapsto \mathbb{R}^{d}$. Then publishing $\mathcal{A}(\mathcal{D})+\mathbf{\kappa}$ preserves $\epsilon$-differential privacy, where $\mathbf{\kappa} \in \mathbb{R}^{d}$ is sampled from the distribution with density
	\begin{equation*}
	p(\mathbf{\kappa}) \propto \exp \left(-\frac{\varepsilon\|\mathbf{\kappa}\|_2}{\Delta_{2}(\mathcal{A})}\right). 
	\end{equation*}
\end{lemma}

\section{Privacy-preserving One-Bit Matrix Completion}\label{secmethod}

This section introduces the proposed framework for privacy-preserving one-bit matrix completion. Corresponding to different stages in one-bit matrix completion, we design four mechanisms i) input, ii) objective iii) gradient, and iv) output perturbation to address the privacy issues. We also provide a theoretical privacy guarantee for each approach. Furthermore, we derive the upper bound of recovery error for the input, objective, and output perturbations and reveal the privacy-accuracy trade-offs.

\subsection{Private Input Perturbation}\label{input}
\begin{algorithm}[t!]
	\caption{Private input perturbation}
	\label{Al1}
	\begin{algorithmic}[1]
		\REQUIRE Incomplete binary matrix $\mathbf{Y}\in \{-1,1\}^{{d_1 \times d_2}}$ ; observation set $\Omega$ ; privacy parameter $\epsilon$.\\

		\STATE Let $\overline{{\mathbf{Y}}}$ be the result of conversion defined in Eq.~\eqref{transform}.
		\STATE Apply SPG to solve Eq.~\eqref{cobj} and obtain $\widehat{\mathbf{M}}$.
		
		\ENSURE Approximate underlying matrix $\widehat{\mathbf{M}}$.
	\end{algorithmic}
\end{algorithm}

In this approach, we exert perturbation on entries of the input matrix $\mathbf{Y}$. Considering that all values in $\mathbf{Y}$ are binary, rather than continuous, we design a perturbation mechanism by changing the original values of entries via a certain probability. Specifically, for each entry $\mathbf{Y}_{ij}= 1(\text{or}\;-1), (i,j) \in \Omega$, it satisfies the following perturbation rules:
\begin{equation}\label{transform}
\overline{\mathbf{Y}}_{ij}=	\left\{
\begin{array}{lll}
-\mathbf{Y}_{ij},&\quad \text{w.p.} &p_1(\operatorname{or}\;p_2)\\
\mathbf{Y}_{ij}, &\quad \text{w.p.} &1-p_1(\operatorname{or}\;1-p_2)
\end{array},\right.
\end{equation}
where $p_1$ and $p_2$ represent the pre-determined transition probabilities. Given $\mathbf{M}$, we can express $\overline{{\mathbf{Y}}}$ by
	\begin{equation*}
		\overline{\mathbf{Y}}_{ij}=\left\{\begin{array}{ll}+1 & \text { w.p. } c\left(\mathbf{M}_{ij}\right) \\ -1 & \text { w.p. } 1-c\left(\mathbf{M}_{ij}\right)\end{array} \right.,
	\end{equation*}
	where $c(x) = h(x)(1-p_1) + [1-h(x)]p_2$ and $h(x)$ is the same as that in Section~\ref{3.2}. To recover $\mathbf{M}$, we solve the following optimization problem: 
	\begin{equation}\label{cobj}
		\underset{\mathbf{X}}{ \min }~ \mathcal{L}_{\Omega, \overline{{\mathbf{Y}}}}(\mathbf{X}), \quad \text{s.t.} ~ \mathbf{X}\in \mathcal{C},
	\end{equation}
where $\mathcal{C} = \{\mathbf{X}: \|\mathbf{X}\|_{*} \leq  \tau$, $\|\mathbf{X}\|_{\infty} \leq  \alpha\}$ and   
\begin{equation*}
  \begin{aligned} 
    \mathcal{L}_{\Omega, \overline{{\mathbf{Y}}}}\left(\mathbf{X}\right) = -\frac{1}{2} & \sum_{(i,j) \in \Omega}\left\{{\left(1 + \overline{\mathbf{Y}}_{ij}\right)\log[c(\mathbf{X}_{ij})]}\right.\\
    &\left.{+ \left(1-\overline{\mathbf{Y}}_{ij}\right)(\log[1-c(\mathbf{X}_{ij})])}\right\}.
  \end{aligned}
\end{equation*}
We solve~\eqref{cobj} via SPG and obtain the estimated matrix $\widehat{\mathbf{M}}$. The summary of this process is shown in Alg.~\ref{Al1}. In addition, we provide theoretical analysis for this perturbation method in Theorem~\ref{Theo1}.

\begin{theorem}\label{Theo1}
 	Given the transition probabilities $p_1$ and $p_2$, Alg.~\ref{Al1} satisfies $\epsilon$-differential privacy if 
	\begin{equation*}
	1 - p_2e^\epsilon \leq  p_1 \leq  (1-p_2)e^\epsilon
	\end{equation*}
	holds. Let $\widehat{\mathbf{M}}$ denote the derived matrix minimizes the optimization problems~\eqref{cobj} and let $\mathbf{M}^\star$ denote the optimal solution. Assuming that $p_1 = p_2 = p$, we have that
	\begin{equation*}
		\epsilon = \log\left(\frac{1-p}{p}\right),
	\end{equation*}
	and with probability at least $1-C_{1} /\left(d_{1}+d_{2}\right)$
	\begin{equation*}
	    \frac{1}{d_{1} d_{2}}\|\widehat{\mathbf{M}}-\mathbf{M}^\star\|_{F}^{2} \leq  C_{c,\alpha} \sqrt{\frac{r\left(d_{1}+d_{2}\right)}{n}} \sqrt{\frac{C_{\Omega,\mathbf{Y}}}{n}},
	\end{equation*}
	where $C_{c,\alpha}=C_{2} \alpha L_{c,\alpha} \beta_{c,\alpha}$, $C_{\Omega,\mathbf{Y}} = \sqrt{n+d_{1} d_{2} \log (d_{1}+d_{2})}$ and $C_1,C_2$ are constants.
\end{theorem}
\begin{proof}
	The detailed proof is given in Appendix~\ref{proof1}.
\end{proof}

When $p_1 = p_2 = p$, we take the transition probability $p = 1/(1+e^\epsilon) \leq 1/2$ that can satisfy $\epsilon$-differential privacy. This indicates that the input perturbation tends not to change the observation matrix $\mathbf{Y}$. Theorem~\ref{Theo1} demonstrates that a larger transition probability encourages a smaller privacy budget $\epsilon$. Given the perturbed observation matrix $\overline{\mathbf{Y}}$, it also provides the recovery error bound that is proportion to $L_{c,\alpha}$ and $\beta_{c,\alpha}$. It is worth noting that the logistic model satisfies Theorem~\ref{Theo1} with $L_{c,\alpha} \leq  (1-2p)$ and $\beta_{c,\alpha} \leq  \frac{1}{2(1-2p)^2}+\frac{1}{2}e^\alpha$, whereas the Gaussian model has $L_{c, \alpha} \leq  \frac{8(1-2p)(\alpha / \sigma+1)}{\sigma}$ and $\beta_{c,\alpha} \leq  \frac{1}{2(1-2p)^2}+\frac{1}{2}\pi\sigma^{2} e^{\alpha^{2} / 2 \sigma^{2}}$.

\subsection{Private Objective Perturbation}
\begin{algorithm}[t!]
	\caption{Private objective perturbation}
	\label{Al4}
	\begin{algorithmic}[1]
		\REQUIRE Incomplete binary matrix $\mathbf{Y}\in\{-1,1\}^{{d_1 \times d_2}}$; observation set $\Omega$ ; privacy parameter $\epsilon$.
		
		\STATE Sample noise matrix $\mathbf{H}=\{\mathbf{H}_{ij}\}$ where $\mathbf{H}_{ij}$ is generated in the way described in Theorem~\ref{the OP}.
		
		\STATE Generate the perturbed objective function $\mathcal{L}^\prime_{\Omega, \mathbf{Y}}(\mathbf{X})$.
		
		\STATE Apply SPG and obtain $\widehat{\mathbf{M}}$.
		
		\ENSURE Approximate underlying matrix $\widehat{\mathbf{M}}$.
	\end{algorithmic}
\end{algorithm}

The basic idea of private objective perturbation for achieving differential privacy is to add random noises to the objective function. Specifically, we propose the following objective function
\begin{equation}\label{pobj}
		\underset{\mathbf{X}}{\min }~ \mathcal{L}^\prime_{\Omega, \mathbf{Y}}(\mathbf{X}), \quad \text{s.t.} ~ \mathbf{X}\in \mathcal{C},
	\end{equation}
where $\mathcal{C} = \{\mathbf{X}: \|\mathbf{X}\|_{*} \leq  \tau$, $\|\mathbf{X}\|_{\infty} \leq  \alpha\}$ and   
\begin{equation*}
  \begin{aligned} 
    \mathcal{L}^\prime_{\Omega, \mathbf{Y}}\left(\mathbf{X}\right) =& -\frac{1}{2}\sum_{(i,j) \in \Omega}\left\{{\left(1 + \mathbf{Y}_{ij}\right)\log[h(\mathbf{X}_{ij})]}\right.\\
    &\left.{+ \left(1-\mathbf{Y}_{ij}\right)(\log[1-h(\mathbf{X}_{ij})])} + \mathbf{H}_{ij}\mathbf{X}_{ij}\right\},
  \end{aligned}
\end{equation*}
where $\{\mathbf{H}_{ij}\}$ are \emph{i.i.d} random noises drawn from exponential distributions. The parameters of the distributions to ensure $\epsilon$-differential privacy are given in the following theorem. After perturbing the objective function, we apply SPG to solve the problem. The whole process is shown in Alg.~\ref{Al4}.

\begin{theorem}\label{the OP}
	Let $\widehat{\mathbf{M}}$ be the derived matrix minimizing the optimization problems~\eqref{pobj} and let ${\mathbf{M}^\star}$ be the optimal solution. Assume that each noise element $\mathbf{H}_{ij}$ is independently and randomly picked from the density function $p\left(\mathbf{H}_{ij}\right) \propto \exp\left({-\frac{\epsilon|\mathbf{H}_{ij}|}{\Delta}}\right),$ where 
	\begin{align}
	\Delta = \underset{i,j}{\max}~\left[\frac{h^\prime(\widehat{\mathbf{M}}_{ij})}{h(\widehat{\mathbf{M}}_{ij})} + \frac{h^\prime(\widehat{\mathbf{M}}_{ij})}{1-h(\widehat{\mathbf{M}}_{ij})}\right]\nonumber.	
	\end{align}
	Then, Alg.~\ref{Al4} satisfies $\epsilon$-differential privacy. Also, with probability at least $(1-C_3/(d_1+d_2))(1 - 1/\sqrt[3]{n})$,
    \begin{equation*}
        \begin{aligned}
            \frac{1}{d_{1} d_{2}}\|\widehat{\mathbf{M}}-{\mathbf{M}^\star}\|_{F}^{2}\leq C_{h,\alpha}\sqrt{\frac{r\left(d_{1}+d_{2}\right)}{n}} \sqrt{\frac{C_{\Omega,\mathbf{Y}}}{n}}+ \frac{C_{h,\alpha}^\prime}{\epsilon\sqrt[3]{n}},
        \end{aligned}
    \end{equation*}
    where $C_{h,\alpha}=C_{4} \alpha L_{h,\alpha} \beta_{h,\alpha}, C_{h,\alpha}^\prime = 8\sqrt{2}\Delta\beta_{h,\alpha}, C_{\Omega,\mathbf{Y}} = \sqrt{n+d_{1} d_{2} \log (d_{1}+d_{2})}$ and $C_3,C_4$ are constants.
\end{theorem}

\begin{proof}
	The detailed proof is shown in Appendix~\ref{proof2}.
\end{proof}
 
 Theorem~\ref{the OP} indicates that the distributions of noises are relevant to the solution of problem~\eqref{pobj}. We can observe that the value of $\Delta$ depends on $\widehat{\mathbf{M}}$, which cannot be achieved in advance. To address this issue, we recall that the solution satisfies 
\begin{equation*}
	\widehat{\mathbf{M}} \in  \{\mathbf{X} : \|\mathbf{X}\|_{*} \leq  \tau, \|\mathbf{X}\|_{\infty} \leq  \alpha\},
\end{equation*}
which means all entries of $\widehat{\mathbf{M}}$ are restricted to the range $[-\alpha, \alpha]$. Thus we can bound $\Delta \leq  2h^\prime(0)/h(-\alpha)$ for Gaussian noise model and $\Delta \leq  1$ for the logistic noise model.

\subsection{Private Gradient Perturbation}

\begin{algorithm}[t!]
	\caption{Private gradient perturbation}
	\label{Al2}
	\begin{algorithmic}[1]
		\REQUIRE Incomplete binary matrix $\mathbf{Y}\in\{-1,1\}^{{d_1 \times d_2}}$; observation set $\Omega$ ; maximum iteration number $K$; spectral step length $\gamma_k, i = 1,2,\cdots,K$; step length $\alpha$, privacy parameter $\epsilon$.\\
		
		\FOR{$K$ iterations}
			\STATE Sample noise matrix $\mathbf{H}=\{ \mathbf{H}_{ij}\}$ where $ \mathbf{H}_{ij}$ is drawn \emph{i.i.d.} from Laplacian distribution with mean $0$ and scale $K\Delta_1(\nabla f(\mathbf{x}))/\epsilon$.
			\STATE $\nabla f(\mathbf{x}_{(k)}) \leftarrow \nabla f(\mathbf{x}_{(k)}) +  \mathcal{P}_\Omega(\mathcal{V}(\mathbf{H}))$.
			\STATE $	\mathbf{x}_{(k+1)}=\mathcal{P}_{\mathcal{C}^\prime}\left(\mathbf{x}_{(k)}-\alpha\gamma_{k} \nabla f(\mathbf{x}_{(k)})\right)$.
		\ENDFOR
		
		\ENSURE Approximate underlying matrix $\widehat{\mathbf{M}}$.
	\end{algorithmic}
\end{algorithm}

Back to problem~\eqref{OBj}, we rewrite it as follows
\begin{equation}\label{objvec}
    \underset{\mathbf{x}}{\min }~ f(\mathbf{x}) \quad \text{s.t.} \quad \mathbf{x} \in \mathcal{C}^\prime,
\end{equation}
where $f(\mathbf{x})$ is a smooth convex function from $\mathbb{R}^{d_1d_2}$ \text{to} $\mathbb{R}$, and $\mathcal{C}^{\prime}$ is a closed convex set in $\mathbb{R}^{d_1d_2}.$ In particular, defining $\mathcal{V}$ to be the bijective linear mapping that vectorizes $\mathbb{R}^{d_{1} \times d_{2}}$ to $\mathbb{R}^{d_{1} d_{2}}$, we have $f(\mathbf{x})=\mathcal{L}_{\Omega, \mathbf{Y}}\left(\mathcal{V}^{-1} \mathbf{x}\right)$, and $\mathcal{C}^\prime$ equals to $\mathcal{V}\left(\mathcal{C}\right)$. We then apply SPG to solve~\eqref{objvec}. Specifically, in the $k$-th iteration, 
\begin{equation*}
	\mathbf{x}_{(k+1)}=\mathcal{P}_{\mathcal{C}^\prime}\left(\mathbf{x}_{(k)}-\alpha\gamma_{k} \nabla f(\mathbf{x}_{(k)})\right),
\end{equation*}
where $\mathbf{x}_{(k)}$ denotes the solution after $k$ iterations, $\gamma_{k}$ denotes the step length calculated according to~\cite{barzilai1988two}, and $\mathcal{P}_{\mathcal{C}^\prime}$ denotes the projection operation as
\begin{equation*}
	\mathcal{P}_{\mathcal{C}^\prime}(\mathbf{v})=  \underset{\mathbf{x}}{\arg \min}~\|\mathbf{x}-\mathbf{v}\|_{2}, \quad \text { s.t.} \quad \mathbf{x} \in \mathcal{C}^\prime.
\end{equation*}

In this approach, we add noises to the gradient of SPG to ensure privacy protections. Specifically, in the $k$-th iteration, $\nabla f(\mathbf{x}_{(k)})$ is perturbed as
\begin{equation*}
	\nabla f^\prime(\mathbf{x}_{(k)}) = \nabla f(\mathbf{x}_{(k)}) + \mathcal{P}_\Omega(\mathcal{V}(\mathbf{H})),
\end{equation*}
where $\mathbf{H}$ denotes the noise matrix. To maintain $\epsilon$-differential privacy, its entries are drawn from Laplacian distribution with mean $0$ and scale $K\Delta_1(\nabla f(\mathbf{x}))/\epsilon$, where $K$ denotes maximum iteration number. Alg.~\ref{Al2} summarizes this process. 
 
\begin{theorem}\label{the GP}
	Alg.~\ref{Al2} satisfies $\epsilon$-differential privacy. 
\end{theorem}

\begin{proof}
	The detailed proof is shown in Appendix~\ref{proof3}.
\end{proof}

The private gradient perturbation in Alg.~\ref{Al2} achieves privacy preservation by adding noises to the gradient of the SPG algorithm. The Laplace mechanism ensures $\epsilon/K$-differential privacy for each iteration. As talked in proof, the sequential composition property ensures that the $K$ iterations maintain the overall $\epsilon$-differential privacy. In practice, we clamp each entry of the gradient to limit the effect of noise. We denote the clamping parameter by $C$~($0.5$ in our experiments) and set it as follows:
\begin{equation*}
    x = \left\{\begin{array}{cc}
         C, & \quad x > C  \\
         x, & \quad -C < x \leq C \\
         -C, & \quad x \leq -C
    \end{array} \right..
\end{equation*}
The clamping reduces the $L_1$-sensitivity of the computations conducted during the SPG process, and therefore results in lower magnitudes of noise being introduced to the differentially private computation.

\subsection{Private Output Perturbation}
In this approach, we still focus on the problem~\eqref{OBj}, which is solved through SPG. The output perturbation maintains a privacy guarantee by adding noises to the output estimation of SPG. The outline of this method is illustrated in Alg.~\ref{Al3}. We note that the distribution of noises is relevant to the $L_1$-sensitivity of the approximate matrix $\widehat{\mathbf{M}}$. Its $L_1$-sensitivity is bounded by the infinite norm constraint, which results in known magnitudes of noises in differentially private computation. 

\begin{algorithm}[t!]
	\caption{Private output perturbation}
	\label{Al3}
	\begin{algorithmic}[1]
		\REQUIRE Incomplete binary matrix $\mathbf{Y}\in\{-1,1\}^{{d_1 \times d_2}}$; observation set $\Omega$ ; privacy parameter $\epsilon$; infinity norm constraint $\alpha$.\\
		
		\STATE Apply SPG to solve Eq.~\eqref{OBj} and obtain $\widehat{\mathbf{M}}$.
		\STATE Sample noise matrix $\mathbf{H}=\{\mathbf{H}_{ij}\}$ where $\mathbf{H}_{ij}$ is \emph{i.i.d} random variable drawn from Laplace distribution with mean $0$ and scale $\frac{2\alpha}{\epsilon}$.
		\STATE $\widehat{\mathbf{M}} \leftarrow \widehat{\mathbf{M}} + \mathbf{H}$.
		
		\ENSURE Approximate underlying matrix $\widehat{\mathbf{M}}$.
	\end{algorithmic}
\end{algorithm}

\begin{theorem}\label{out}
	Alg.~\ref{Al3} maintains $\epsilon$-differential privacy. Furthermore, let $\widehat{\mathbf{M}}$ be the derived matrix of Alg.~\ref{Al3} and let ${\mathbf{M}^\star}$ be the optimal solution of problem~\eqref{OBj}. With the same constants $C_{\Omega,\mathbf{Y}}$ and $C_{f,\alpha}$ in Theorem \ref{Theo1}, we have that with probability at least $1-C_5/ (d_1+d_2)$,
    \begin{equation*}
        \begin{aligned}
        \frac{1}{d_1d_2}\|\widehat{\mathbf{M}}-\mathbf{M}^\star\|_F^2 \leq & C_{f,\alpha}\sqrt{\frac{r\left(d_{1}+d_{2}\right)}{n}} \sqrt{\frac{C_{\Omega,\mathbf{Y}}}{n}}\\
        &+ \frac{C_6\alpha^2\log^2(d_1+d_2) }{\epsilon^2},
        \end{aligned}
	\end{equation*}
	where $C_5, C_6$ are positive constants.
\end{theorem}

\begin{proof}
	The detailed proof is left to Appendix~\ref{proof4}.
\end{proof}

Theorem~\ref{out} provides an $\epsilon$-differential privacy guarantee and the recovery bound of the proposed algorithm. To ensure the convergence of the error bound, it's worth noting that we require $\log(d_1+d_2) = \operatorname{o}(\epsilon)$. 

\begin{figure*}[t!]
	\subfigure[$\epsilon$-differential privacy with Gaussian noise]{
		\begin{minipage}[t]{0.42\linewidth}			
			\includegraphics[width=1\linewidth]{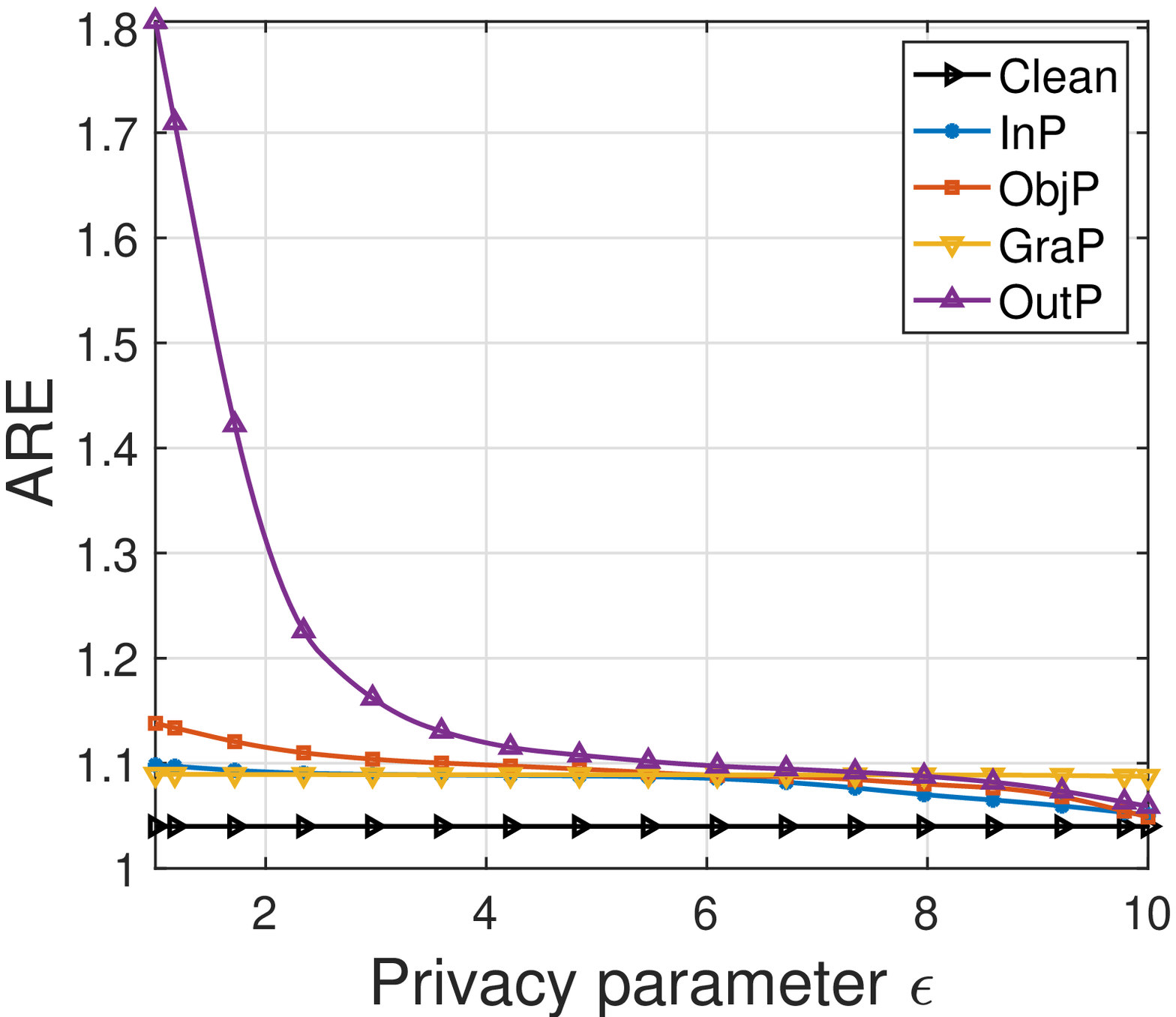}
		\end{minipage}
	}
	\subfigure[$\epsilon$-differential privacy with Logistic noise]{
		\begin{minipage}[t]{0.42\linewidth}
				\includegraphics[width=1\linewidth]{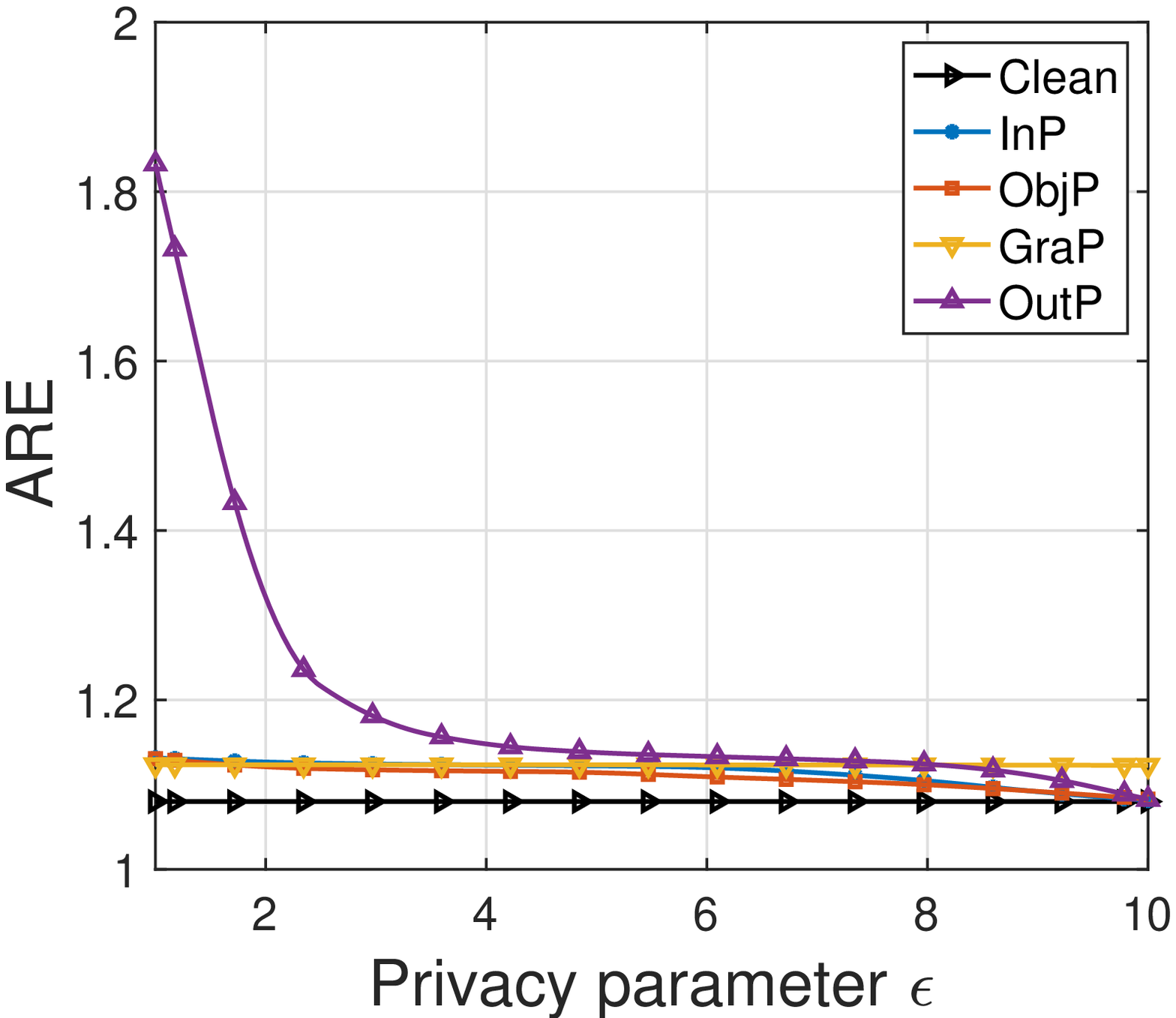}
		\end{minipage}%
	}
	\centering
	\caption{Recovery performance on synthetic data with different privacy parameter $\epsilon$}
	\label{Fig1}
\end{figure*}

\section{Experiments}\label{secsim}

In this section, we evaluate the performance of our proposed framework of privacy-preserving one-bit matrix completion. Section~\ref{Setting} describes the datasets and setting used in the following experiments. In Section~\ref{5.2}, we show the recovery performance on all the datasets. And Section~\ref{5.3} discusses the effect of observation ratios.

\subsection{Datasets and Experimental Settings}\label{Setting}
The overview of datasets is shown as follows. We conduct experiments on both synthetic data and real-world datasets. The latter are summarized in Table~\ref{TB1}.

\begin{figure*}[t!]
	\subfigure[$\epsilon$-differential privacy with Gaussian noise]{
		\begin{minipage}[t]{0.42\linewidth}			
			\includegraphics[width=1\linewidth]{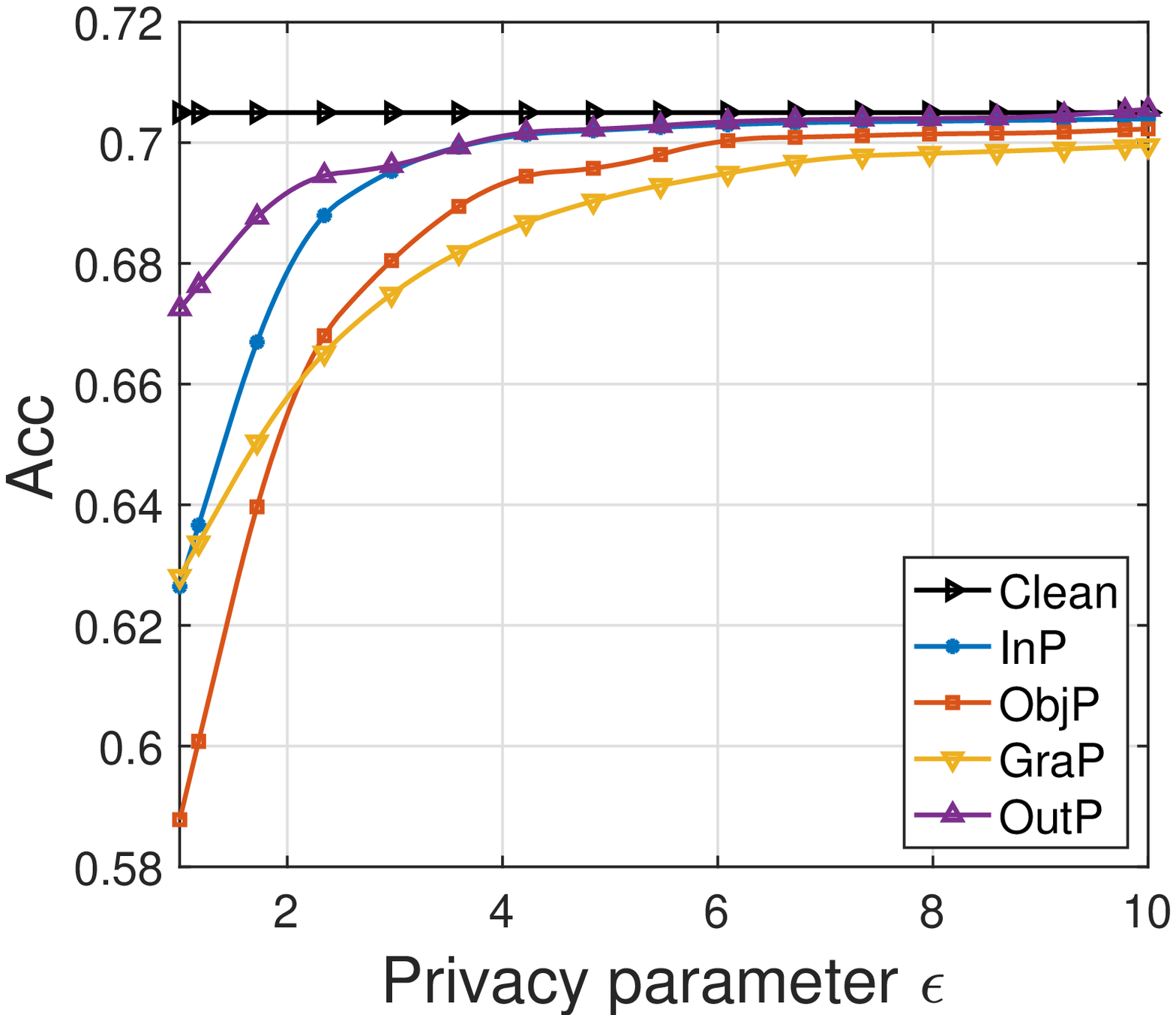}
		\end{minipage}
	}
	\subfigure[$\epsilon$-differential privacy with Logistic noise]{
		\begin{minipage}[t]{0.42\linewidth}
				\includegraphics[width=1\linewidth]{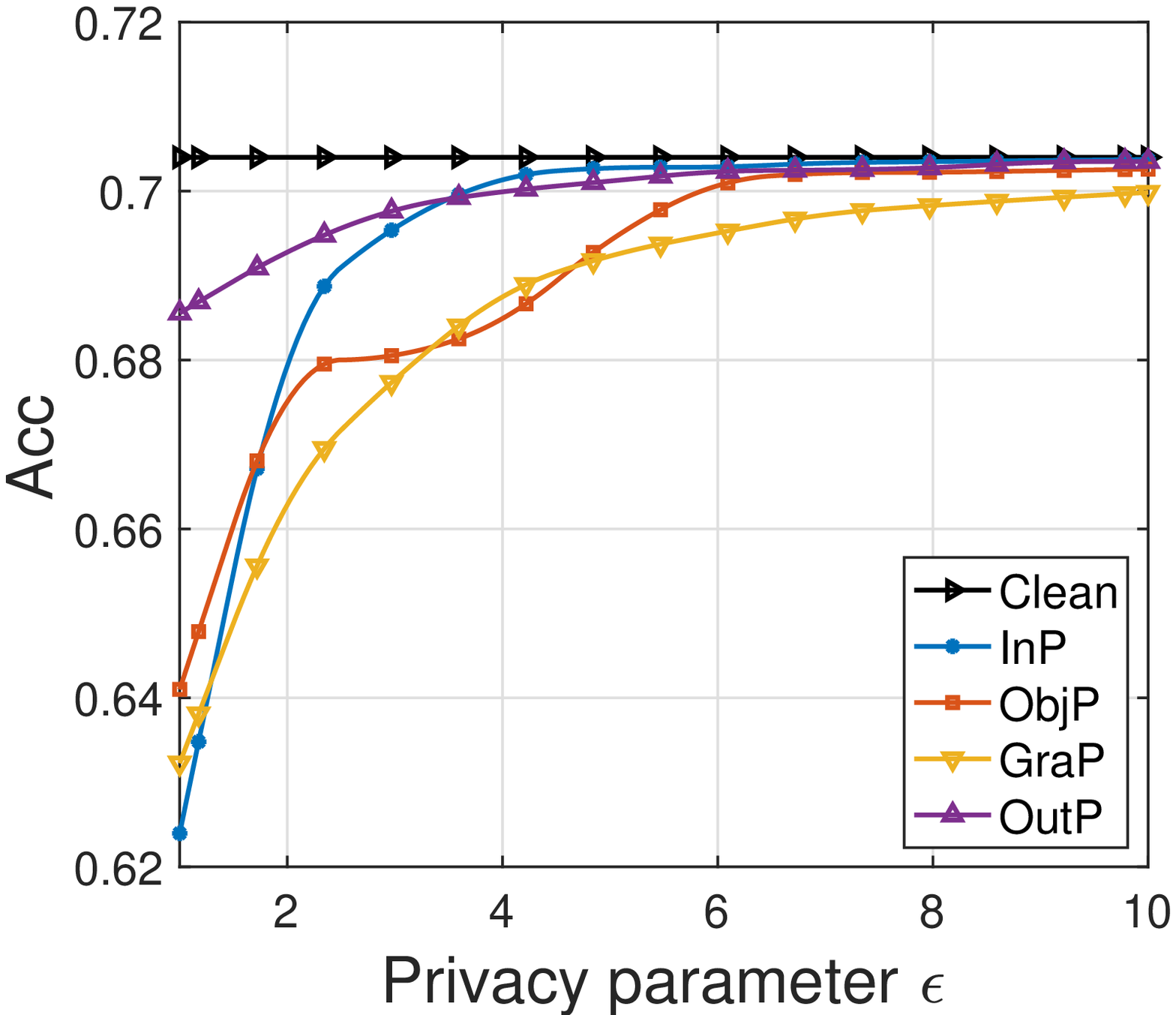}
		\end{minipage}%
	}
	\centering
	\caption{Recovery performance on ML-100K with different privacy parameter $\epsilon$}
	\label{Fig2}
 
\end{figure*}
\begin{figure*}[t!]
	\centering
	\subfigure[$\epsilon$-differential privacy with Gaussian noise]{
		\begin{minipage}[t]{0.42\linewidth}			
			\includegraphics[width=1\linewidth]{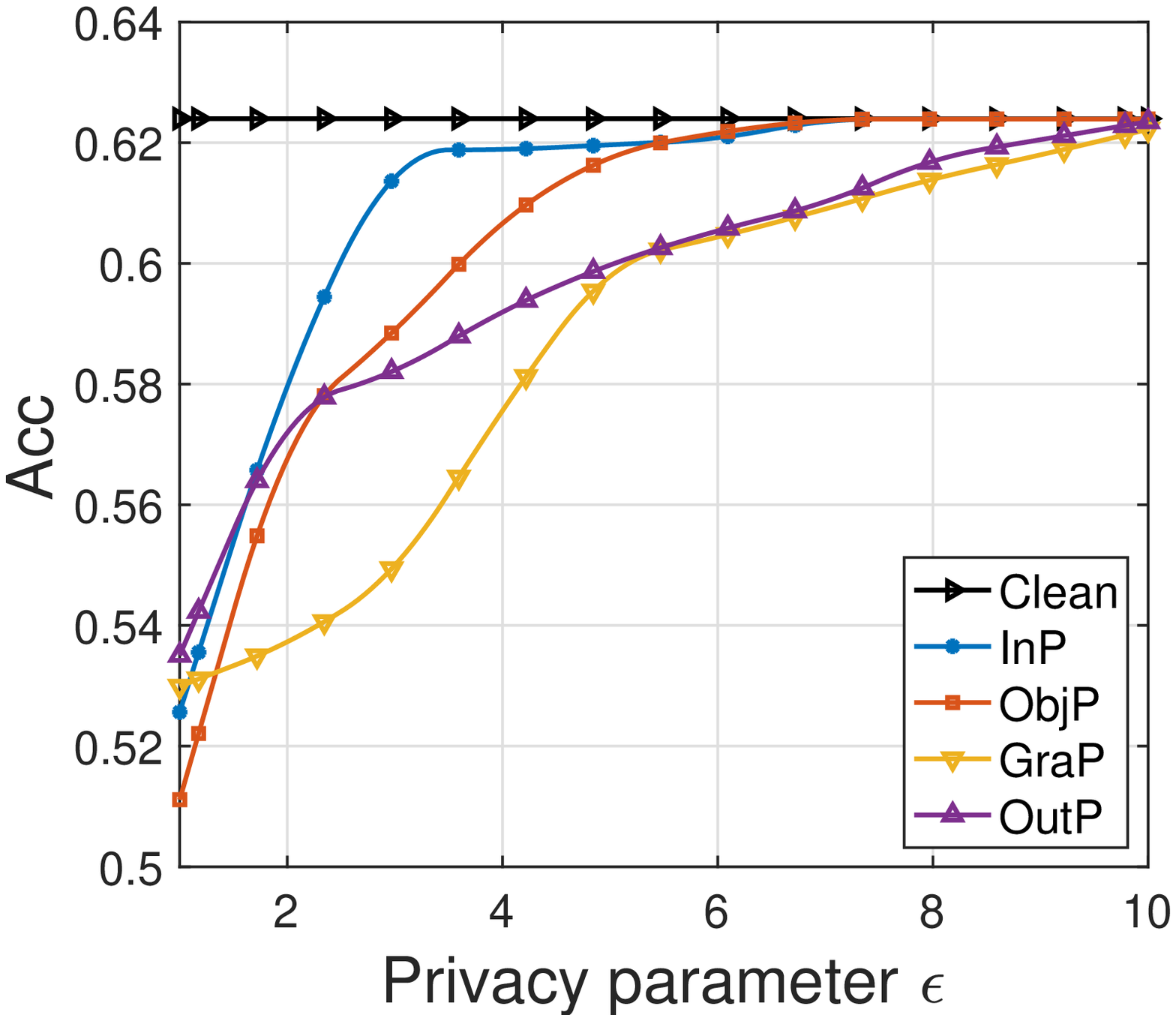}
		\end{minipage}
	}
	\subfigure[$\epsilon$-differential privacy with logistic noise]{
		\begin{minipage}[t]{0.42\linewidth}
				\includegraphics[width=1\linewidth]{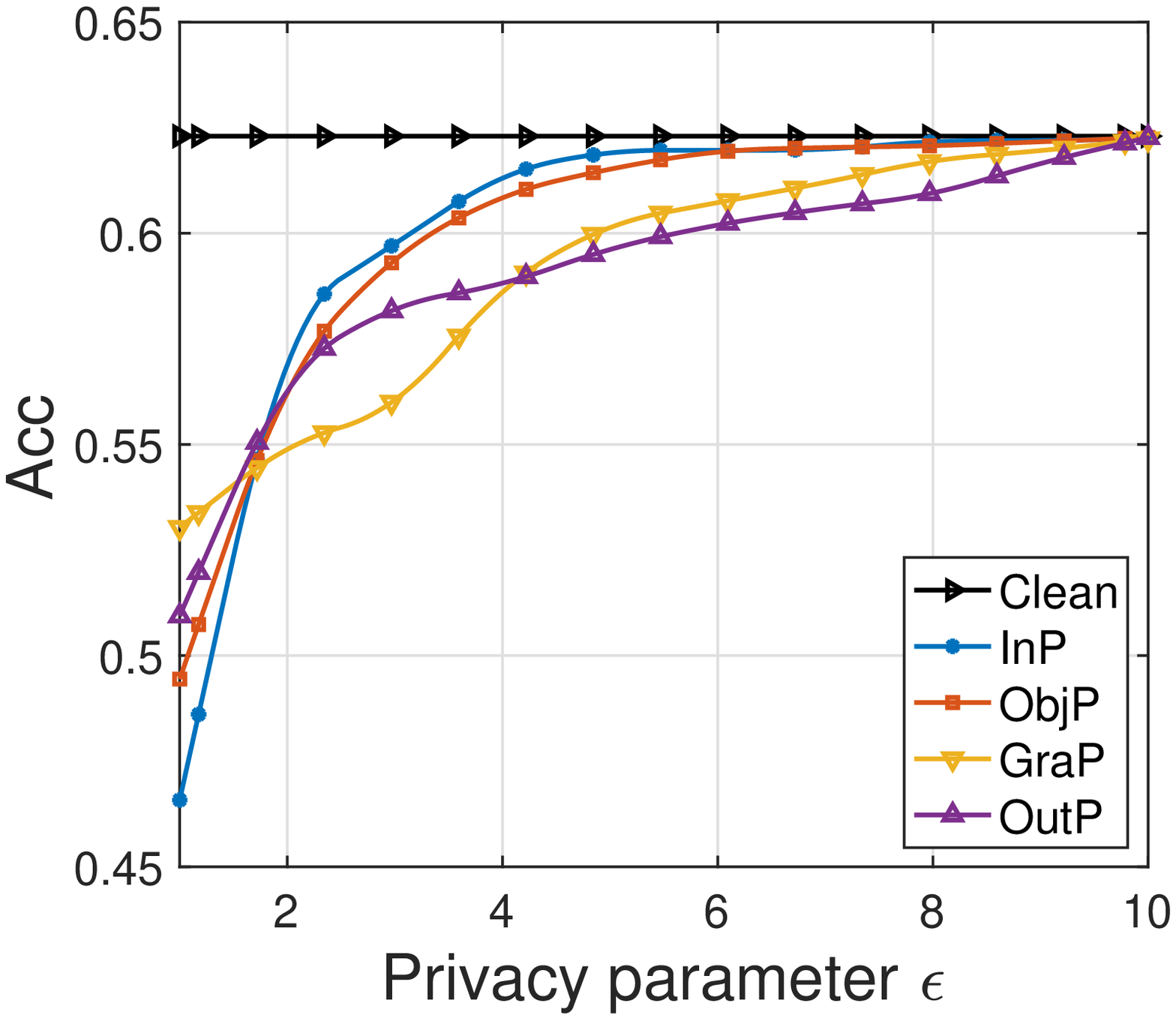}
		\end{minipage}%
	}
	\centering
	\caption{Recovery performance on RC with different privacy parameter $\epsilon$}
	\label{Fig3}
\end{figure*}

{\bf 1. Synthetic data} We construct the underlying matrix $\mathbf{M} \in \mathbb{R}^{d_1 \times d_2}$ with rank-$r$ by $\mathbf{M} =\mathbf{M}_1 \mathbf{M}_2^\top$, where all entries of $\mathbf{M}_1$ and $\mathbf{M}_2$ are drawn from the uniform distribution on the range of $[-1/2, 1/2]$ and $\mathbf{M}_1\in \mathbb{R}^{d_1\times r}, \mathbf{M}_2 \in \mathbb{R}^{d_2\times r}$. To satisfy the max-norm constraint, we scale $\mathbf{M}$ as
	\begin{equation*}
		\mathbf{M}_{ij} = \frac{\alpha \mathbf{M}_{ij}}{\max_{i,j} \mathbf{M}_{ij}}
	\end{equation*}
	such that $\|\mathbf{M}\|_{\infty} = \alpha$. Afterwards, we introduce Gaussian and Logistic noises to generate target binary matrix $\mathbf{Y}$ via Eq.~\eqref{generateYY} with observation ratio $15\%$, respectively. The scale of Gaussian noises is $1$. We set $r = \alpha = 1$ and $d_1 = d_2 = 100$. The radius of nuclear-norm ball is set to $\tau = \alpha  \sqrt{d_1 d_2 r}$. 

{\bf 2. MovieLens 100K~(ML-100K)~\cite{harper2015movielens}} The dataset contains 100,000 ratings from 943 users across 1,682 films, with each rating occurring on a scale from 1 to 5. We convert these ratings to binary observations for testing purposes by comparing each rating to the average rating for the entire dataset (which is approximately 3.5). Ratings below the average rating are set as $-1$, while those above the average rating are set as $1$. We use the default split provided for training and testing in the original dataset. The data are available for download at \url{http//www.grouplens.org/node/73}.

{\bf 3. RC~\cite{beckerleg2020divide}} The dataset contains customer reviews from138 customers for 130 restaurants, in which ratings are in $\{0, 1, 2\}$.
 We map them to binary data by converting two-star ratings to $1$ and converting the zero- or one-star ratings to $-1$. Then we randomly split RC data into the training set and testing set with a ratio of $8:2$. The data are available for download at \url{https//www.kaggle.com/uciml/restaurant-data-with-consumer-ratings}.
\begin{table}[h!]
	\renewcommand\arraystretch{1.3}
	\normalsize
	\centering
	\caption{Description of real-world datasets} 
	\begin{tabular}{lcccc}
		\hline
	    Datasets  &  Users  &  Items  &  Ratings  &  Density  \\
	    \hline
	    \hline
	    ML-100K  &  $943$  &  $1682$  &  $\{1,2,3,4,5\}$  &  $6.3\%$  \\
	    \hline
	    RC      &   $138$  &  $130$   &  $\{0,1,2\}$   &  $5.2\%$   \\
	    \hline
	\end{tabular}
	\label{TB1}
\end{table}

For the simulated data, we use the average relative error~(ARE) to evaluate the performance
\begin{equation*}
	\operatorname{ARE} = \frac{\|\widehat{\mathbf{M}} - \mathbf{M}\|_F^2}{\|\mathbf{M}\|_F^2}.
\end{equation*}

For real data, however, since the scoring matrix of the real data does not have a corresponding implicit probability matrix $\mathbf{M}$, there is no ``ground truth'' that can be used to measure the recovery accuracy. Thus, we evaluate the overall effect by comparing whether the sign of the estimated value of $\widehat{\mathbf{M}}$ accurately predicts the sign of the remaining test set. That is, we use the correct match accuracy (Acc) to evaluate the performance:
\begin{equation*}
	\operatorname{Acc} = \frac{\#_{(i,j)\in \text{test data}} \{ \operatorname{sgn}(\widehat{\mathbf{M}}_{ij}) = \operatorname{sgn}(\mathbf{Y}_{ij}) \}}{\#\{\text{test data}\}},
\end{equation*}
where $\text{sgn}(x)$ denotes the symbol function and $\#\{\cdot\}$ denotes the number of set $\{\cdot\}$. 

We conduct four perturbation mechanisms to ensure differential privacy on all datasets with logistic noise and Gauss noise. We denote input, objective, gradient, and output perturbation by InP, ObjP, GraP, and OutP respectively. For comparison, we apply SPG without any perturbation, denoted by Clear, as the baseline. We consider the range of privacy parameter $\epsilon \in [1,10]$. We notice that small $\epsilon$ can impair the recovery performance while maintaining a high level of privacy-preserving. On the contrary, a large $\epsilon$ provides weak privacy protection, whereas matters very little to the recovery performance. Therefore, we can obtain an efficient and flexible tradeoff between the recovery performance and level of privacy protection.

\subsection{Recovery Performance on Different Datasets}\label{5.2}

\subsubsection{Performance on Synthetic Data}
We average the results over $40$ draws of $\widehat{\mathbf{M}}$. Fig.~\ref{Fig1} shows the performance of three privacy-deserving mechanisms with both Gaussian and logistic noises. In general, as the privacy allowed to leak gradually increases, the recovery errors of the three mechanisms under both noises gradually decrease and approach Clear. ObjP, GraP, and InP performance is relatively stable, which means that they are less sensitive to the changes in privacy parameter $\epsilon$. In contrast, OutP performs worse when $\epsilon$ goes smaller because it directly adds noises to the reconstructed matrix $\widehat{\mathbf{M}}$ obtained by the SPG algorithm. As $\epsilon$ becomes larger~(e.g., $\epsilon = 4$), the performance of OutP is comparable to the other methods.

\begin{figure*}[t!]
	\setlength{\abovecaptionskip}{-0.2cm}
	\includegraphics[scale = 0.32]{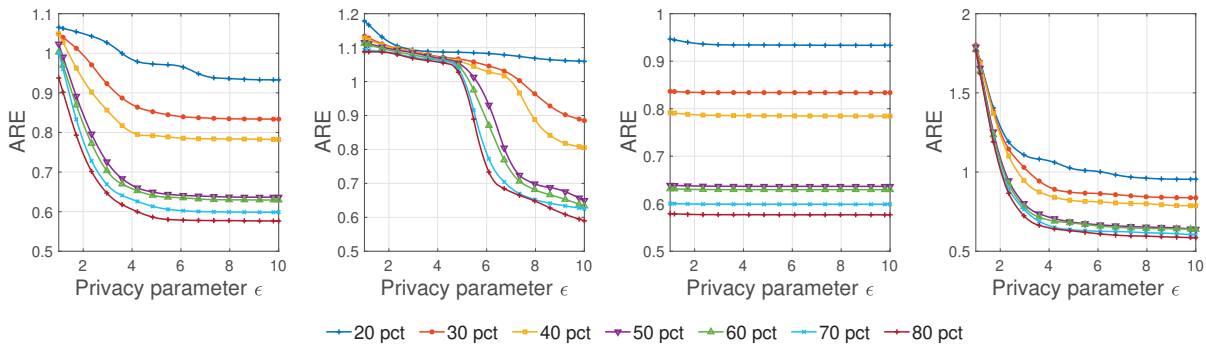}
	\centering
	\caption{Recovery performance with Gauss noises under different observation ratios}
	\label{Fig4}
\end{figure*}

\begin{figure*}[t!]
	\setlength{\abovecaptionskip}{-0.2cm}
	\includegraphics[scale = 0.32]{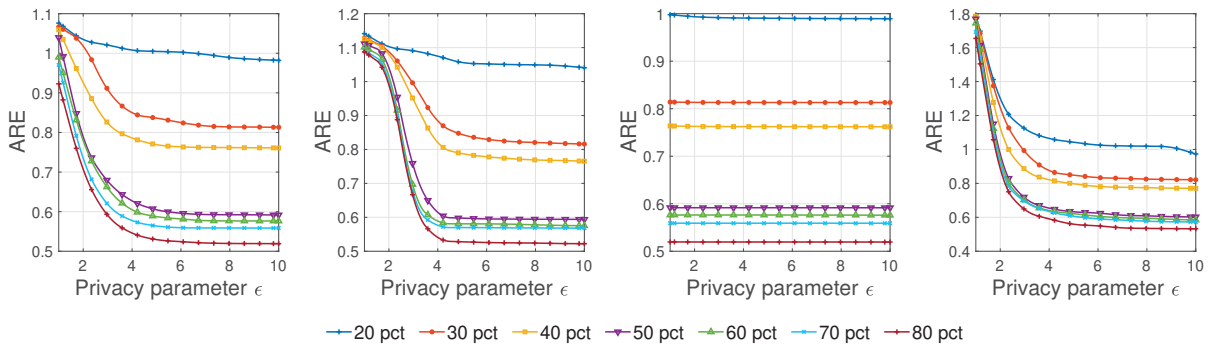}
	\centering
	\caption{Recovery performance with logistic noises under different observation ratios}
	\label{Fig5}
\end{figure*}

\subsubsection{Performance on Real-World Datasets}
We repeat the experiment on real-world datasets $10$ times and average these results in Fig.~\ref{Fig2}. The results show that OutP achieves the best results among the four perturbation methods. When the parameters $\epsilon$ are relatively small, the other three methods are less effective, which is different from our synthetic results. This is because for the real-world datasets, we cannot obtain the probability matrix $\mathbf{M}$ and the specific probability function used to generate the observation matrix $\mathbf{Y}$, which could influence the effect of InP to a certain extent. For ObjP and GraP, since the magnitude of noises for privacy protection depends on the probability distribution, this also leads to inferior results for small $\epsilon$. In addition, the prediction performances of our proposed four mechanisms all become better as the privacy parameter $\epsilon$ goes larger, and gradually approaches that of the Clear approach without differential privacy. Fig.~\ref{Fig2} shows that at approximately $\epsilon=4$, all the four perturbation mechanisms achieve the accuracy larger than $68\%$.

Fig.~\ref{Fig3} shows the prediction performance on the RC dataset for all perturbation mechanisms. The results are also averaged over $10$ experiments. As with ML-100K, we observe that there are some differences between the results on the RC dataset and the synthetic results. When $\epsilon$ is relatively small~(e.g., $\epsilon = 1$), the results of InP, ObjP, and GraP are unsatisfactory. This again illustrates the effect of knowing the specific distribution function for generating the observation matrix. Overall, these mechanisms can obtain the expected trade-off between the privacy and accuracy, which adds validity to our proposed framework in practical scenarios.

\subsection{Recovery Performance with Different Observation Ratios}\label{5.3}
To evaluate the effect of different observation ratios on our proposed perturbation mechanisms, we conduct experiments with observation ratios ranging from $20$ to $80$ with both Gaussian and logistic noises. For each case, we average the results over $10$ experiments. Since one-bit matrix completion relies on the observed entries to estimate the unknown part, the observation ratio can be viewed as training size. Fig.~\ref{Fig4} and~\ref{Fig5} indicate that as the observation ratio increases, the performance of all four mechanisms becomes better. For InP and OutP, ARE first decreases as the parameter $\epsilon$ becomes more significant, but levels off after approximately $\epsilon = 5$. This is consistent with our theoretical analysis since InP and OutP do not directly affect the estimation process of the SPG algorithm. When the level of added noises gradually decreases, the influence on the observation matrix or the recovery matrix also gradually decreases. In contrast, the recovery performance of ObjP changes drastically with increasing $\epsilon$. Interestingly, when the privacy parameter $\epsilon \leq  1$, the performances of ObjP and OutP with both types of noises are less prone to change with the observation ratio. For ObjP, adding noises to the objective function can affect the convergence of the algorithm, which induces bad performance regardless of the training size. OutP influences the estimation of the SPG algorithm directly. When the parameter $\epsilon$ is small and the added noises are relatively large, this perturbation will still destroy the final estimation result, although increasing the observation ratio can get a reasonable estimate.

\section{Conclusion}\label{seccon}
One-bit matrix completion has a wide variety of applications in recommendation systems, yet suffering from privacy concerns. To address this issue, we have applied differential privacy and proposed a unified privacy-preserving framework. Our framework contains four perturbation mechanisms that work in different stages. For each mechanism, we have provided a theoretical guarantee that ensures $\epsilon$-differential privacy for our proposal. Furthermore, we have analyzed the recovery error bound and give provable results. Experiments on both synthetic and real-world data have demonstrated that our proposed mechanisms yield gratifying performance while maintaining high-level privacy.

This work offers essential insights for future studies. First of all, our framework collects private data from users and assumes the server that conducts SPG is trusted. When the server is untrusted, however, there exists a risk of information leakage. This case is worthy of future investigation. Second, we can apply our proposal to provide privacy protection for several important variants of one-bit matrix completion~\cite{cai2013max, bhaskar20151, ni2016optimal}. Lastly, we are interested in extending our proposal to the scenarios where data are presented by high dimensional tensor consisting of binary values.

\appendices
\section{Proof of Theorem~\ref{Theo1}}\label{proof1}
\begin{proof}
	Denote the input mechanism as $\mathcal{A} : \mathcal{D} \mapsto \mathbb{R}^d$ and $\mathbf{Y}, \mathbf{Y}^\prime \in \mathcal{D}$ represent two sets of ratings differing only at one record, i.e., $\mathbf{Y}_{pq}$, ${\mathbf{Y}}^\prime_{pq}$. Without loss of generality, we assume $\mathbf{Y}_{pq} = -{\mathbf{Y}}^\prime_{pq} = 1$. Considering that the input mechanism can change the original input to its counterpart by a certain probability, we can leverage it to deduce the probability ratio of having the same output region for these two datasets. For one output region $S$, if ${\mathbf{Y}}^\prime_{pq}$ has been changed to $1$, then
	\begin{equation*}
	\begin{split}
	\frac{\mathbb{P}[\mathcal{A}(\mathbf{Y}) \in S]}{\mathbb{P}[\mathcal{A}(\mathbf{Y}^\prime) \in S]} = \frac{1-p_1}{p_2} \leq  \frac{p_2e^\epsilon}{p_2} = e^\epsilon;
	\end{split}
	\end{equation*}
If ${\mathbf{Y}}_{pq}$ has been changed to $-1$, then	\begin{equation*}
	\begin{split}
	\frac{\text{P}[\mathcal{A}(\mathbf{Y}) \in S]}{\text{P}[\mathcal{A}(\mathbf{Y}^\prime) \in S]} = \frac{p_1}{1-p_2} \leq  \frac{(1-p_2)e^\epsilon}{1-p_2} = e^\epsilon.
	\end{split}
	\end{equation*}
Thus, the input mechanism satisfies $\epsilon$-differential privacy. 

As for building the recovery error bound, we follow the proof of~\cite[Theorem 1]{davenport20141}. It is worth noting that we focus on a new probability function defined as 
	\begin{equation*}
		c(x) = h(x)(1-p) + [1-h(x)]p,
	\end{equation*}
	 where $h(x)$ is defined in Section~\ref{3.2}. Our proof contains mainly two parts. Recall the definitions
	\begin{equation*}
		L_{c,\alpha}=\sup _{|x| \leq  \alpha} \frac{\left|c^{\prime}(x)\right|}{c(x)(1-c(x))}, \beta_{c,\alpha}=\sup _{|x| \leq  \alpha} \frac{c(x)(1-c(x))}{\left(c^{\prime}(x)\right)^2}.
	\end{equation*}
 	We first derive the upper bounds of $L_{c,\alpha}$ and $\beta_{c,\alpha}$, respectively. Then we show that both functions
	\begin{equation*}
		\frac{1}{L_{c,\gamma}} \log \left(\frac{c(x)}{c(0)}\right) \text{ and } \frac{1}{L_{c,\gamma}} \log \left(\frac{1-c(x)}{1-c(0)}\right)
	\end{equation*}
	are contractions that vanish at $0$.
	
	\begin{enumerate}
	\item[i)] {\bf The upper bounds of $L_{c,\alpha}$ and $\beta_{c,\alpha}$.} 	
	
	By the definition of $c(x)$ and noticing $0 < p = \frac{1}{1+e^\epsilon} < \frac{1}{2}$, we have
	\begin{equation*}
		\begin{split}
		&c^\prime(x) = (1-2p)h^\prime(x),\\
		&c(x)(1-c(x)) = p(1-p)\left[h^2(x)+(1-h(x))^2\right] + h(x) ((1-h(x))\left[p^2+(1-p)^2\right] \\
			 	&	\qquad \qquad \qquad	~~\geq h(x)((1-h(x)).
		\end{split}
	\end{equation*}
	Thus, we can bound $L_{c,\alpha}$ as
	\begin{equation*}
		L_{c,\alpha} \leq  \sup _{|x| \leq  \alpha}\frac{(1-2p)|h^\prime(x)|}{h(x)((1-h(x))} = (1-2p)L_{h,\alpha}.
	\end{equation*}
	Note that the logistic model satisfies $L_{h,\alpha}=1$ and Gaussian model has $L_{h, \alpha} \leq  8 \frac{\alpha / \sigma+1}{\sigma}$, which induces the upper bound of $L_{c,\alpha}$. For $\beta_{c,\alpha}$, we derive that 
	\begin{eqnarray*}
				 \left(c^\prime(x)\right)^2 & \hspace{-1mm} = & \hspace{-1mm}  (1-2p)^2\left(h^\prime(x)\right)^2,\\
			c(x)(1-c(x)) & \hspace{-1mm} = & \hspace{-1mm} p(1-p)\left[h^2(x)+(1-h(x))^2\right] +  h(x) ((1-h(x))\left[p^2+(1-p)^2\right] \\
			 		     & \hspace{-1mm} \leq & \hspace{-1mm} \frac{1}{2}+\frac{1}{2}(1-2p)^2h(x)(1-h(x)),
	\end{eqnarray*}
	where the last equation is due to Cauchy-Schwartz inequality. Thus, we have
	\begin{equation*}
		\begin{split}
			\beta_{c,\alpha} &\leq  \sup _{|x| \leq  \alpha} \frac{\frac{1}{2}+\frac{1}{2}(1-2p)^2h(x)(1-h(x))}{(1-2p)^2\left(h^{\prime}(x)\right)^2} = \frac{1}{2(1-2p)^2}+\frac{1}{2}\beta_{h,\alpha}.
		\end{split}
	\end{equation*}
	Note that the logistic model satisfies $\beta_{h,\alpha} = e^\alpha$ and Gaussian model has $\beta_{h, \alpha} \leq  \pi\sigma^{2} e^{\alpha^{2} / 2 \sigma^{2}}$, which induces the upper bound of $\beta_{c,\alpha}$.

	\item[ii)]{\bf The contraction property.}
	
	We first notice that $c(x) \in (0,1), \text{ when } x\in [-\alpha, \alpha]$. Thus both functions
	\begin{equation*}
		\frac{1}{L_{c,\gamma}} \log \left(\frac{c(x)}{c(0)}\right) \text{ and } \frac{1}{L_{c,\gamma}} \log \left(\frac{1-c(x)}{1-c(0)}\right)
	\end{equation*}
	are well defined. Then,
	\begin{align}
		\left|\frac{1}{L_{c,\gamma}} \log \left(\frac{c(x)}{c(0)}\right) - \frac{1}{L_{c,\gamma}} \log \left(\frac{c(y)}{c(0)}\right)\right|&= \frac{1}{L_{c,\gamma}}\left| \log \left(c(x)\right) - \log \left(c(y)\right)\right|\nonumber \\
		&= \frac{1}{L_{c,\gamma}}\left| \frac{c^\prime(\xi_1)}{c(\xi_1)}\right|\left| x-y\right|\label{Lcc},
	\end{align}
	where $\xi_1 \in (x,y)$ and the last equation is due to Mean Value Theorem. By the definition of $L_{c,\gamma}$, \eqref{Lcc} turns to
	\begin{equation*}
		\begin{split}
		\left|\frac{1}{L_{c,\gamma}} \log \left(\frac{c(x)}{c(0)}\right) - \frac{1}{L_{c,\gamma}} \log \left(\frac{c(y)}{c(0)}\right)\right|&= \frac{|c^\prime(\xi_1)|}{c(\xi_1)(1-c(\xi_1))L_{c,\gamma}}(1-c(\xi_1))|x-y| \nonumber \\		
		&\leq   (1-c(\xi_1)) |x-y| \nonumber\\
		&\leq  |x-y|\nonumber.
		\end{split}
	\end{equation*}

	Similarly, we have 
	\begin{align}
		\left|\frac{1}{L_{c,\gamma}} \log \left(\frac{1-c(x)}{1-c(0)}\right) - \frac{1}{L_{c,\gamma}}\log \left(\frac{1-c(y)}{1-c(0)}\right)\right|&= \frac{1}{L_{c,\gamma}}\left| \log \left(1-c(x)\right) - \log \left(1-c(y)\right)\right|\nonumber \\
		&= \frac{1}{L_{c,\gamma}}\left| \frac{c^\prime(\xi_2)}{1-c(\xi_2)}\right|\left| x-y\right| \nonumber \\
		&= \frac{|c^\prime(\xi_1)|}{c(\xi_1)(1-c(\xi_1))L_{c,\gamma}}c(\xi_1)|x-y| \nonumber \\		
		&\leq   c(\xi_1)|x-y| \nonumber\\
		&\leq  |x-y|\nonumber.
	\end{align}
	Thus both functions are contractions. 
	\end{enumerate}
	
	Finally, following the analysis of~\cite[Theorem 1]{davenport20141}, we complete our proof.
\end{proof}

\section{Proof of Theorem~\ref{the OP}}\label{proof2}
	\begin{proof}
		Suppose that $\mathbf{Y}$ and $\mathbf{Y}^\prime$ are two sets of ratings differing at one record $\mathbf{Y}_{pq}$ and $\mathbf{Y}^\prime_{pq}$. Denote $\mathbf{H} = \{\mathbf{H}_{ij}\}$ and $\mathbf{H}^\prime = \{\mathbf{H}^\prime_{ij}\}$ as the noise matrices for $\mathbf{Y}$ and $\mathbf{Y}^\prime$, respectively. Since the objective function is convex, we have that $\forall (i,j) \in \Omega$, 
	\begin{equation*}
		\nabla \mathcal{L}^\prime_{\Omega, \mathbf{Y}}(\widehat{\mathbf{M}})=\nabla \mathcal{L}^\prime_{\Omega, \mathbf{Y}^\prime}(\widehat{\mathbf{M}})=0,
	\end{equation*}
	which indicates 
	\begin{equation}\label{gradient}
    \begin{split}
      \mathbf{H}_{ij} -(1+\mathbf{Y}_{ij})\frac{h^\prime(\widehat{\mathbf{M}}_{ij})}{2h(\widehat{\mathbf{M}}_{ij})}+(1-\mathbf{Y}_{ij})\frac{h^\prime(\widehat{\mathbf{M}}_{ij})}{2(1-h(\widehat{\mathbf{M}}_{ij}))}  &= \mathbf{H}_{ij}^\prime -(1+{\mathbf{Y}}_{ij})\frac{h^\prime(\widehat{\mathbf{M}}_{ij})}{2h(\widehat{\mathbf{M}}_{ij})}  \\&+ (1-\mathbf{Y}^\prime_{ij})\frac{h^\prime(\widehat{\mathbf{M}}_{ij})}{2(1-h(\widehat{\mathbf{M}}_{ij}))}.
    \end{split}
	\end{equation}
Without loss of generality, suppose that $\mathbf{Y}_{pq} = -\mathbf{Y}^\prime_{pq} = 1$. If $(i,j) \neq (p,q)$, then $\mathbf{Y}_{ij} = \mathbf{Y}_{ij}^\prime$ and we can derive from~\eqref{gradient} that $ \mathbf{H}_{ij} =  \mathbf{H}_{ij}^\prime$. 

If $(i, j) = (p, q)$, we can obtain from~\eqref{gradient} that
	\begin{equation*}
	 \mathbf{H}_{pq} -  \mathbf{H}_{pq}^\prime = \frac{h^\prime(\widehat{\mathbf{M}}_{pq})}{h(\widehat{\mathbf{M}}_{pq})} + \frac{h^\prime(\widehat{\mathbf{M}}_{pq})}{1-h(\widehat{\mathbf{M}}_{pq})}, 
	\end{equation*}
which means that $\forall (i,j) $,
	\begin{equation*}
	| \mathbf{H}_{pq} -  \mathbf{H}_{pq}^\prime| \leq  \underset{i,j}{\max}~\left[\frac{h^\prime(\widehat{\mathbf{M}}_{ij})}{h(\widehat{\mathbf{M}}_{ij})} + \frac{h^\prime(\widehat{\mathbf{M}}_{ij})}{1-h(\widehat{\mathbf{M}}_{ij})}\right] :=\Delta.
	\end{equation*}		
Therefore, for any pair of $ \mathbf{Y}_{p q}$ and $\mathbf{Y}^\prime_{pq}$,
	\begin{equation*}
	\begin{split}
	 \frac{\mathbb{P}\left[\mathbf{M} = \widehat{\mathbf{M}} \mid \mathbf{Y}\right]}{\mathbb{P}\left[\mathbf{M} = \widehat{\mathbf{M}} \mid \mathbf{Y}^\prime\right]}
	=&  \frac{\prod_{i\in [d_1], j \in [d_2]} p\left( \mathbf{H}_{ij}\right)}{\prod_{i\in [d_1], j \in [d_2]} p\left( \mathbf{H}_{ij}^{\prime}\right)}  \\
	=&  \exp  \left\{  -  \frac{\epsilon  \left(\sum_{i\in [d_1], j \in [d_2]}| \mathbf{H}_{ij}| -  \sum_{i\in [d_1], j \in [d_2]}| \mathbf{H}_{ij}^{\prime}|\right)}{ \Delta}\right\}\\
	=& \exp\left\{-\frac{\epsilon\left(| \mathbf{H}_{pq}|-| \mathbf{H}_{pq}^{\prime}|\right)}{ \Delta}\right\}\\
	\leq  & e^\epsilon.\\
	\end{split}
	\end{equation*}
Thus, Alg.~\ref{Al4} satisfies $\epsilon$-differential privacy.

Now we derive the recovery error bound. Note that 
\begin{equation*}
  \begin{aligned} 
    \mathcal{L}^\prime_{\Omega, \mathbf{Y}}\left(\mathbf{X}\right) &= -\frac{1}{2}\sum_{(i,j) \in \Omega}\left\{{\left(1 +  \mathbf{Y}_{ij}\right)\log[h(\mathbf{X}_{ij})]}{+ \left(1- \mathbf{Y}_{ij}\right)(\log[1-h(\mathbf{X}_{ij})])} +  \mathbf{H}_{ij}\mathbf{X}_{ij}\right\}\\
    &= \mathcal{L}_{\Omega, \mathbf{Y}}\left(\mathbf{X}\right) + \sum_{(i,j) \in \Omega} \mathbf{H}_{ij}\mathbf{X}_{ij}.
  \end{aligned}
\end{equation*}
Denote 
\begin{equation*}
	\begin{split}
		\overline{\mathcal{L}}_{\Omega, \mathbf{Y}}(\mathbf{X})&=\mathcal{L}_{\Omega, \mathbf{Y}}(\mathbf{X})-\mathcal{L}_{\Omega, \mathbf{Y}}(\mathbf{0}),\\
		\overline{\mathcal{L}^\prime}_{\Omega, \mathbf{Y}}(\mathbf{X})
		&=\mathcal{L}^\prime_{\Omega, \mathbf{Y}}(\mathbf{X})-\mathcal{L}^\prime_{\Omega, \mathbf{Y}}(\mathbf{0}).\\
	\end{split}
\end{equation*}
Notice that $\mathcal{L}_{\Omega, \mathbf{Y}}(\mathbf{0}) = \mathcal{L}^\prime_{\Omega, \mathbf{Y}}(\mathbf{0})$. Thus,
\begin{equation*}
	\begin{split}
		\overline{\mathcal{L}^\prime}_{\Omega, \mathbf{Y}}(\mathbf{X})&=\mathcal{L}^\prime_{\Omega, \mathbf{Y}}(\mathbf{X})-\mathcal{L}^\prime_{\Omega, \mathbf{Y}}(\mathbf{0}) \\
		&=\mathcal{L}^\prime_{\Omega, \mathbf{Y}}(\mathbf{X}) - \mathcal{L}_{\Omega, \mathbf{Y}}(\mathbf{0}) \\
		&= \mathcal{L}_{\Omega, \mathbf{Y}}\left(\mathbf{X}\right) + \sum_{(i,j) \in \Omega} \mathbf{H}_{ij}\mathbf{X}_{ij}- \mathcal{L}_{\Omega, \mathbf{Y}}(\mathbf{0}) \\
		&= \overline{\mathcal{L}}_{\Omega, \mathbf{Y}}(\mathbf{X}) + \sum_{(i,j) \in \Omega} \mathbf{H}_{ij}\mathbf{X}_{ij}.
	\end{split}
\end{equation*}
For any $\mathbf{X},\mathbf{M} \in \mathcal{C}$, we obtain
\begin{equation}\label{theo2-1}
	\begin{split}
		 &\quad ~\overline{\mathcal{L}^\prime}_{\Omega, \mathbf{Y}}(\mathbf{X}) - \overline{\mathcal{L}^\prime}_{\Omega, \mathbf{Y}}(\mathbf{M})=  \overline{\mathcal{L}}_{\Omega, \mathbf{Y}}\left(\mathbf{X}\right) -  \overline{\mathcal{L}}_{\Omega, \mathbf{Y}}\left(\mathbf{M}\right) + \sum_{(i,j) \in \Omega} \mathbf{H}_{ij}(\mathbf{X}_{ij}- \mathbf{M}_{ij}).
	\end{split}
\end{equation}
According to~\cite[Theorem 1]{davenport20141}, we have
\begin{equation*}
	\begin{aligned}
   	&\quad~ \overline{\mathcal{L}}_{\Omega, \mathbf{Y}}(\mathbf{X}) - \overline{\mathcal{L}}_{\Omega, \mathbf{Y}}(\mathbf{M})\leq  2\sup _{\mathbf{X} \in \mathcal{C}}\left|\overline{\mathcal{L}}_{\Omega, \mathbf{Y}}(\mathbf{X})-\mathrm{E}\left[\overline{\mathcal{L}}_{\Omega, \mathbf{Y}}(\mathbf{X})\right]\right| -n \operatorname{D}(h(\mathbf{M}) \| h(\mathbf{X})),
\end{aligned}
\end{equation*}
where $\operatorname{D}(\cdot \| \cdot)$ denotes the KL divergence of two distributions. Then, Eq.~\eqref{theo2-1} turns to
\begin{equation*}
	\begin{aligned}
	\overline{\mathcal{L}^\prime}_{\Omega, \mathbf{Y}}(\mathbf{X}) - \overline{\mathcal{L}^\prime}_{\Omega, \mathbf{Y}}(\mathbf{M}) &\leq  2\sup _{\mathbf{X} \in \mathcal{C}}\left|\overline{\mathcal{L}}_{\Omega, \mathbf{Y}}(\mathbf{X})-\mathrm{E}\left[\overline{\mathcal{L}}_{\Omega, \mathbf{Y}}(\mathbf{X})\right]\right| -n \operatorname{D}(h(\mathbf{M}) \| h(\mathbf{X})) \\&+ \sum_{(i,j) \in \Omega} \mathbf{H}_{ij}(\mathbf{X}_{ij}- \mathbf{M}_{ij})\nonumber.
\end{aligned}
\end{equation*}
Since $\widehat{\mathbf{M}}$ is the solution to problem~\eqref{pobj}, we also have that for $\forall \mathbf{X} \in \mathcal{C}$, $\overline{\mathcal{L}^\prime}_{\Omega, \mathbf{Y}}(\mathbf{X}) \geq \overline{\mathcal{L}^\prime}_{\Omega, \mathbf{Y}}(\widehat{\mathbf{M}})$. Thus
\begin{align}
	&0 \leq  2\sup _{\mathbf{X} \in \mathcal{C}}\left|\overline{\mathcal{L}}_{\Omega, \mathbf{Y}}(\mathbf{X})-\mathrm{E}\left[\overline{\mathcal{L}}_{\Omega, \mathbf{Y}}(\mathbf{X})\right]\right| -n \operatorname{D}(h(\widehat{\mathbf{M}}) \| h(\mathbf{X}))+ \sum_{(i,j) \in \Omega} \mathbf{H}_{ij}(\mathbf{X}_{ij}- \widehat{\mathbf{M}}_{ij})\nonumber,
\end{align}
which means
\begin{align}\label{theo2-2}
	\operatorname{D}(h(\widehat{\mathbf{M}}) \| h(\mathbf{X}))& \leq  \frac{2}{n}\sup _{\mathbf{X} \in \mathcal{C}}\left|\overline{\mathcal{L}}_{\Omega, \mathbf{Y}}(\mathbf{X})-\mathrm{E}\left[\overline{\mathcal{L}}_{\Omega, \mathbf{Y}}(\mathbf{X})\right]\right| + \frac{1}{n}\sum_{(i,j) \in \Omega} \mathbf{H}_{ij}(\mathbf{X}_{ij} -\widehat{\mathbf{M}}_{ij}).
\end{align}
% We focus on bounding the two terms in the right side of Eq.~\eqref{theo2-2} respectively. For the first term, 
We apply~\cite[Lemma 1]{davenport20141} to bound the first term in the right-hand side of Eq.~\eqref{theo2-2}. Specifically, with probability at least $1-C_3/(d_1+d_2)$,
\begin{align}\label{theo2-3}
	& \quad~\frac{1}{n}\sup _{\mathbf{X} \in \mathcal{C}}\left|\overline{\mathcal{L}}_{\Omega, \mathbf{Y}}(\mathbf{X})-\mathbb{E} \overline{\mathcal{L}}_{\Omega, \mathbf{Y}}(\mathbf{X})\right|\leq C_4\alpha L_{h,\gamma} \sqrt{\frac{r\left(d_{1}+d_{2}\right)}{n}} \sqrt{\frac{C_{\Omega,\mathbf{Y}}}{n}},
\end{align}
where $C_{\Omega, \mathbf{Y}} = \sqrt{n+d_{1} d_{2} \log \left(d_{1} d_{2}\right)}$ and $C_3, C_4$ are constants. Considering the constraint $\|\mathbf{M}\|_{\infty} \leq \alpha$, we can bound the second term in the right-hand side of Eq.~\eqref{theo2-2} by
\begin{equation*}
    \begin{aligned}
    \frac{1}{n}\sum_{(i, j) \in \Omega} \mathbf{H}_{ij}(\mathbf{X}_{ij} - \widehat{\mathbf{M}}_{ij}) \leq \frac{2\alpha}{n}  \sum_{(i, j) \in \Omega} \mathbf{H}_{ij},
    \end{aligned}
\end{equation*}
where $\mathbf{H}_{ij}$ has the distribution with density $p( \mathbf{H} ) \propto \exp \left(-\frac{\varepsilon| \mathbf{H}|}{\Delta}\right)$ and $\Delta$ is defined in Theorem~\ref{the OP}. Each $ \mathbf{H}_{ij}$ can be independently generated by $ \mathbf{H}_{ij} = \mathbf{L}_{ij}\mathbf{V}_{ij}$, where $\mathbf{L}_{ij}$ satisfies Gamma distribution $\Gamma(1, \Delta / \epsilon)$ and $\mathbf{V}_{ij}$ is sampled in $\{+1, -1\}$ with the same probability. Therefore, we have
\begin{equation*}
    \begin{aligned}
    \mathbb{E}\left[ \frac{1}{n}\sum_{(i, j)\in \Omega} \mathbf{H}_{ij}\right] = \frac{1}{n}\left(\sum_{(i, j)\in \Omega}(\mathbb{E}[\mathbf{V}_{ij}]\mathbb{E}[\mathbf{L}_{ij}])\right) = 0,
    \end{aligned}
\end{equation*}
and
\begin{equation*}
    \begin{aligned}
    \mathrm{Var}\left[ \frac{1}{n}\sum_{(i, j)\in \Omega} \mathbf{H}_{ij}\right] = \frac{1}{n^2}\left(\sum_{(i, j)\in \Omega}\mathrm{Var}[ \mathbf{H}_{ij}] \right) = \frac{1}{n}\mathrm{Var}[ \mathbf{H}_{ij}].
    \end{aligned}
\end{equation*}
Considering the independence of $\mathbf{L}_{ij}$ and $\mathbf{V}_{ij}$, we have
\begin{equation*}
\begin{aligned}
\mathrm{Var}[ \mathbf{H}_{ij}] = \frac{2\Delta^2}{\epsilon^2}.
\end{aligned}
\end{equation*}
Thus, we can obtain
\begin{equation*}
    \begin{aligned}
    \mathrm{Var}\left[ \frac{1}{n}\sum_{(i, j)\in \Omega} \mathbf{H}_{ij}\right] = \frac{2\Delta^2}{n\epsilon^2}.
    \end{aligned}
\end{equation*}
Applying Chebyshev's inequality we arrive at
\begin{equation*}
    \begin{aligned}
    \mathbb{P}\left(\left| \frac{1}{n}\sum_{(i, j)  \in \Omega}  \mathbf{H}_{ij}\right| < a \right) \geq 1 - \frac{2\Delta^2}{n\epsilon^2a^2},~\forall a>0.
    \end{aligned}
\end{equation*}
If $a$ takes the value of $\sqrt{2}\Delta/\epsilon\sqrt[3]{n}$, we have
\begin{equation}\label{theo2-4}
    \mathbb{P}\left(\left| \frac{1}{n}\sum_{(i, j)\in \Omega}  \mathbf{H}_{ij}\right| < \frac{\sqrt{2}\Delta}{\epsilon\sqrt[3]{n}} \right) \geq 1 - \frac{1}{\sqrt[3]{n}}.
\end{equation}
Due to Eq.~\eqref{theo2-3} and Eq.~\eqref{theo2-4}, we have that with probability at least $(1-C_3/(d_1+d_2))(1 - 1/\sqrt[3]{n})$,
\begin{equation*}
	\operatorname{D}(h(\widehat{\mathbf{M}}) \| h(\mathbf{X})) 
	\leq  C_4\alpha L_{h,\gamma} \sqrt{\frac{r\left(d_{1}+d_{2}\right)}{n}} \sqrt{\frac{C_{\Omega,\mathbf{Y}}}{n}} + \frac{\sqrt{2}\Delta}{\epsilon\sqrt[3]{n}}.
\end{equation*}
Considering the fact that the KL divergence is always greater than Hellinger distance, we have that with probability at least $(1-C_3/(d_1+d_2))(1 - 1/\sqrt[3]{n})$,
\begin{equation}\label{theo2-5}
    \operatorname{d}_H^2(h(\widehat{\mathbf{M}}), h(\mathbf{X})) \leq  C_4\alpha L_{h,\gamma} \sqrt{\frac{r\left(d_{1}+d_{2}\right)}{n}} \sqrt{\frac{C_{\Omega,\mathbf{Y}}}{n}} + \frac{\sqrt{2}\Delta}{\epsilon\sqrt[3]{n}}.
\end{equation}
Applying~\cite[Lemma 2]{davenport20141}, we have
\begin{equation*}
	d_{H}^{2}(h(\widehat{\mathbf{M}}), h(\mathbf{X})) \geq \frac{1}{8\beta_{h, \alpha}}\frac{\|\mathbf{M}-\widehat{\mathbf{M}}\|_{F}^{2}}{d_{1} d_{2}}.
\end{equation*}
Combining this with~\eqref{theo2-5}, we obtain that with probability at least $(1-C_3/(d_1+d_2))(1 - 1/\sqrt[3]{n})$,
\begin{equation*}
\begin{aligned}
    \frac{1}{d_{1} d_{2}}\|\widehat{\mathbf{M}}-\mathbf{M}^\star\|_{F}^{2}\leq C_{h,\alpha}\sqrt{\frac{r\left(d_{1}+d_{2}\right)}{n}} \sqrt{\frac{C_{\Omega,\mathbf{Y}}}{n}}+ \frac{C_{h,\alpha}^\prime}{\epsilon\sqrt[3]{n}},
    \end{aligned}
\end{equation*}
where $C_{h,\alpha}=C_{4} \alpha L_{h,\alpha} \beta_{h,\alpha}, C_{h,\alpha}^\prime = 8\sqrt{2}\Delta\beta_{h,\alpha}, C_{\Omega,\mathbf{Y}} = \sqrt{n+d_{1} d_{2} \log (d_{1}+d_{2})}$ and $C_3,C_4$ are constants.
\end{proof}

\section{Proof of Theorem~\ref{the GP}}\label{proof3}
\begin{proof} 
We finish the proof by the following lemma.
\begin{lemma}[Sequential Composition~\cite{mcsherry2009differentially}]
	\label{le3}
    Let the randomized algorithms $\mathcal{A}_i$ that apply to the dataset $\mathcal{X}$ each provide $\epsilon_i$-differential privacy. Then the sequence of $\mathcal{A}_i(\mathcal{X})$ provides $\sum_i\epsilon_i$-differential privacy.
	\end{lemma}
	According to the Laplacian mechanism, each gradient perturbation preserves $\epsilon/k$-differential privacy. Because the overall number of iterations is $k$, the sequential composition lemma ensures that the whole algorithm maintains $\epsilon$-differential privacy.
\end{proof}

\section{Proof of Theorem~\ref{out}}\label{proof4}

\begin{proof}
	The global sensitivity of the output $\widehat{\mathbf{M}}$ is $\Delta_1 = 2\alpha$. According to Laplace mechanism, this algorithm maintains $\epsilon$-differential privacy. Furthermore, denote the approximate underlying matrix before perturbation as $\widehat{\mathbf{M}}^{(0)}$. We have
	\begin{equation*}
	    \frac{1}{d_1d_2}\|\widehat{\mathbf{M}}-\mathbf{M}^\star\|_F^2 \leq  \frac{1}{d_1d_2}\|\widehat{\mathbf{M}}-\widehat{\mathbf{M}}^{(0)}\|_F^2 + \frac{1}{d_1d_2}\|\widehat{\mathbf{M}}^{(0)}-\mathbf{M}^\star\|_F^2.
	\end{equation*}
	The upper bound of $\|\widehat{\mathbf{M}}^{(0)}-\mathbf{M}^\star\|_F^2$ has been provided in~\cite[Theorem 1]{davenport20141}. Here, We focus only on the bound of $\|\widehat{\mathbf{M}}-\widehat{\mathbf{M}}^{(0)}\|_F^2$.
	Notice that
	\begin{equation*}
	    \|\widehat{\mathbf{M}}-\widehat{\mathbf{M}}^{(0)}\|_F^2 = \|\mathbf{H}\|_F^2 = \sum_{i,j} \mathbf{H}_{ij}^2.
	\end{equation*}
    Since $ \mathbf{H}_{ij}$ is drawn from Laplace distribution with mean $0$ and scale $\frac{2\alpha}{\epsilon}$,  we have that $\forall t>0$,
    \begin{align}
        \mathbb{P}\left( | \mathbf{H}_{ij}| \geq  \frac{2t\alpha}{\epsilon}\right) & = 2\mathbb{P}\left(  \mathbf{H}_{ij} \geq  \frac{2t\alpha}{\epsilon}\right) \nonumber \\
        &=2 \int_{\frac{2t\alpha}{\epsilon}}^{\infty} \frac{\epsilon}{4\alpha} \exp \left(-\frac{\epsilon| \mathbf{H}_{ij}|}{2\alpha}\right) \mathrm{d} \mathbf{H}_{ij} \nonumber \\
        & = e^{-t} \nonumber.
    \end{align}
    Thus,
    \begin{align}
        \mathbb{P}\left( \sum_{i,j} \mathbf{H}_{ij}^2 \geq  \frac{4d_1d_2t^2\alpha^2}{\epsilon^2}\right) & \leq  \mathbb{P}\left( \min_{i,j} \mathbf{H}_{ij}^2 \geq  \frac{4t^2\alpha^2}{\epsilon^2}\right) \nonumber \\
         &= 1 - \left[1-\mathbb{P}\left( | \mathbf{H}_{ij}| \geq  \frac{2t\alpha}{\epsilon}\right)\right]^{d_1d_2} \nonumber \\
         & = 1 - \left[1-e^{-t}\right]^{d_1d_2} \nonumber\\
        & \leq  \frac{d_1d_2}{e^{t}}. \label{talo}
    \end{align}
    Also, let $C_{\Omega,\mathbf{Y}} = \sqrt{n+(d_{1}+ d_{2}) \log (d_{1}d_{2})}$. Together with~\eqref{talo} and~\cite[Theorem 1]{davenport20141}, we have that with probability at least $1-C_5/ (d_1+d_2)$, 
    \begin{equation*}
        \begin{aligned}
        \frac{1}{d_1d_2}\|\widehat{\mathbf{M}}-\mathbf{M}^\star\|_F^2 &\leq  C_{f,\alpha}\sqrt{\frac{r\left(d_{1}+d_{2}\right)}{n}} \sqrt{\frac{C_{\Omega,\mathbf{Y}}}{n}}+ \frac{4\alpha^2\left[ \log(d_1d_2)+C_6\log(d_1+d_2) \right]^2}{\epsilon^2},
        \end{aligned}
	\end{equation*}
	where $C_{f,\alpha}=C_{6} \alpha L_{f,\alpha} \beta_{f,\alpha}$, and $C_5, C_6$ are constants.
\end{proof}

\end{document}